\documentclass[11pt]{article}
\usepackage{tikz}

\usepackage{amsfonts}
\usepackage{amsmath}
\usepackage{amssymb}
\usepackage{enumerate}
\usepackage{xspace}

\newcommand{\CCfont}[1]{\ensuremath{\mathsf{#1}}}

\newcommand{\myset}[2]{ \left\{ #1 \left| #2 \right.\right\} }
\newcommand{\prefix}{\sqsubseteq}
\newcommand{\limsupn}{\limsup\limits_{n\to\infty}}
\newcommand{\limsupN}{\limsup\limits_{N\to\infty}}
\newcommand{\ith}{{i^{\mathrm{th}}}}
\newcommand{\PR}[2]{\underset{#1}{\CCfont{Pr}}\left[#2\right]}
\newcommand{\restr}{\negthinspace\upharpoonright\negthinspace}

\newcommand{\N}{\mathbb{N}}
\newcommand{\Q}{\mathbb{Q}}

\newcommand{\p}{{\CCfont{p}}}
\newcommand{\ptwo}{{\p_{\thinspace\negthinspace_2}}}
\newcommand{\pspace}{{\CCfont{pspace}}}
\newcommand{\ptwospace}{{\CCfont{p}_{\thinspace\negthinspace_2}\CCfont{space}}}
\newcommand{\mup}{\mu_\p}
\newcommand{\muptwo}{\mu_\ptwo}

\newcommand{\C}{\CCfont{C}}

\newcommand{\co}[1]{\CCfont{co}#1}
\newcommand{\DTIME}{\CCfont{DTIME}}
\newcommand{\NTIME}{\CCfont{NTIME}}
\newcommand{\NP}{\CCfont{NP}}
\newcommand{\coNP}{\co{\NP}}
\newcommand{\FewP}{\CCfont{FewP}}
\newcommand{\UP}{\CCfont{UP}}
\newcommand{\coUP}{\co{\UP}}
\newcommand{\AM}{\CCfont{AM}}
\newcommand{\coAM}{\co{\AM}}
\newcommand{\BQP}{\CCfont{BQP}}
\newcommand{\BPP}{\CCfont{BPP}}
\newcommand{\RP}{\CCfont{RP}}
\newcommand{\E}{\CCfont{E}}

\newcommand{\ALMOST}{\CCfont{ALMOST}}

\renewcommand{\P}{\CCfont{P}}
\newcommand{\EXP}{\CCfont{EXP}}

\newcommand{\poly}{\CCfont{poly}}
\newcommand{\NPcoNP}{\NP\cap\coNP}
\newcommand{\AMcoAM}{\AM \cap \coAM}
\newcommand{\UPcoUP}{\UP \cap \coUP}

\newcommand{\calC}{{\cal C}}
\newcommand{\calF}{{\cal F}}
\newcommand{\calG}{{\cal G}}
\newcommand{\calM}{{\cal M}}

\newcommand{\T}{\CCfont{T}}
\newcommand{\leqp}{\leq^\p}
\newcommand{\leqpT}{\leqp_\T}
\renewcommand{\Pr}{\CCfont{P}_r}

\newcommand{\TM}{\CCfont{TM}}
\newcommand{\QTM}{\CCfont{QTM}}

\newcommand{\free}{\CCfont{free}}

\def\binaryn{\lbrace 0,1 \rbrace ^ n}
\def\binary{\lbrace 0,1 \rbrace}
\def\Pr{\CCfont{Pr}}

\def\binaryinfty{ \lbrace 0,1 \rbrace ^ \infty}

\newcommand{\permEXP}{\CCfont{PermEXP}}
\newcommand{\permE}{\CCfont{PermE}}

\renewcommand{\ALMOST}[1]{\CCfont{ALMOST}\textrm{-}#1}
\renewcommand{\ALMOST}[1]{\CCfont{ALMOST}\textrm{-}[#1]}
\renewcommand{\ALMOST}[1]{\CCfont{ALMOST}[#1]}

\newcommand{\BQTIME}{\CCfont{BQTIME}}

\newcommand{\NLIN}{\CCfont{NLIN}}

\newcommand{\emptylist}{[\hspace{.5mm}]}

\newcommand{\bfPi}{\mathsf{\Pi}}
\newcommand{\PFPi}{\CCfont{PP}\bfPi}
\newcommand{\OPPi}{\CCfont{OP}\bfPi}

\renewcommand{\C}{\CCfont{C}}
\newcommand{\cyl}[1]{\llbracket #1 \rrbracket}

\usepackage[numbers,sort&compress]{natbib}
\usepackage{hyperref}

\usepackage[margin=1in]{geometry}
\usepackage{mathtools}
\usepackage{amsmath}
\usepackage{stmaryrd}

\usepackage{environ}
\NewEnviron{appendixproof}[1][]{\begin{proof} \BODY \end{proof}}

\usepackage{amsthm}
\newtheorem{theorem}{Theorem}[section]
\newenvironment{theorem_cite}[1]{\begin{theorem} {\rm (#1)}}{\end{theorem}}
\usepackage{aliascnt}
\usepackage[nameinlink]{cleveref}
\newaliascnt{lemma}{theorem}
\newtheorem{lemma}[lemma]{Lemma}
\aliascntresetthe{lemma}
\crefname{lemma}{lemma}{lemmas}
\newaliascnt{corollary}{theorem}
\newtheorem{corollary}[corollary]{Corollary}
\aliascntresetthe{corollary}
\crefname{corollary}{corollary}{corollaries}
\newaliascnt{proposition}{theorem}
\newtheorem{proposition}[proposition]{Proposition}
\aliascntresetthe{proposition}
\crefname{proposition}{proposition}{propositions}
\newaliascnt{question}{theorem}
\newtheorem{question}[question]{Question}
\aliascntresetthe{question}
\crefname{question}{question}{questions}
\theoremstyle{definition}
\newaliascnt{definition}{theorem}
\newtheorem{definition}[definition]{Definition}
\aliascntresetthe{definition}
\crefname{definition}{Definition}{Definitions}
\Crefname{definition}{Definition}{Definitions}

\crefalias{definition_cite}{definition}
\newaliascnt{remark}{theorem}

\aliascntresetthe{remark}
\crefname{remark}{remark}{remarks}
\usepackage{thm-restate}

\usepackage{paralist}
\newenvironment{enumerateC}{\begin{enumerate}}
{\end{enumerate}}
   \numberwithin{theorem}{section}
\numberwithin{figure}{section}

\newcommand{\HALFRANGE}{\textrm{HALFRANGE}}
\renewcommand{\HALFRANGE}{\CCfont{HRNG}}

\title{\bf Random Permutations in Computational Complexity\footnote{A preliminary version of this paper appeared the {\em Proceedings of the 50th International Symposium on Mathematical Foundations of Computer Science (MFCS 2025)} \cite{Hitchcock:RPCC}.}}

\author{John M. Hitchcock
\thanks{Department of Electrical Engineering and Computer Science, University of Wyoming. jhitchco@uwyo.edu. This research was supported in part by NSF grant 2431657.}
\and Adewale Sekoni
\thanks{Department of Mathematics, Computer Science \& Physics,
  Roanoke College. sekoni@roanoke.edu}
\and
Hadi Shafei
\thanks{Department of Computer Science, University of Huddersfield. H.Shafei@hud.ac.uk}
}

\begin{document}

\maketitle
\begin{abstract}

	Classical results of Bennett and Gill (1981) show that with probability 1, $\P^A \neq \NP^A$ relative to a random oracle $A$, and with probability 1, $\P^\pi \neq \NP^\pi \cap \coNP^\pi$ relative to a random permutation $\pi$.
	Whether $\P^A = \NP^A \cap \coNP^A$ holds relative to a random oracle $A$ remains open.
	While the random oracle separation has been extended to specific individually random oracles--such as Martin-Löf random or resource-bounded random oracles--no analogous result is known for individually random permutations.

	We introduce a new resource-bounded measure framework for analyzing individually random permutations. We define permutation martingales and permutation betting games that characterize measure-zero sets in the space of permutations, enabling formal definitions of polynomial-time random permutations, polynomial-time betting-game random permutations, and polynomial-space random permutations.

	Our main result shows that $\P^\pi \neq \NP^\pi \cap \coNP^\pi$ for every polynomial-time betting-game random permutation $\pi$. This is the first separation result relative to individually random permutations, rather than an almost-everywhere separation.
	We also strengthen a quantum separation of Bennett, Bernstein, Brassard, and Vazirani (1997) by showing that $\NP^\pi \cap \coNP^\pi \not\subseteq \BQP^\pi$ for every polynomial-space random permutation $\pi$.

	We investigate the relationship between random permutations and random oracles. We prove that random oracles are polynomial-time reducible from random permutations. The converse--whether every random permutation is reducible from a random oracle--remains open. We show that if $\NP \cap \coNP$ is not a measurable subset of $\EXP$, then $\P^A \neq \NP^A \cap \coNP^A$ holds with probability 1 relative to a random oracle $A$. Conversely, establishing this random oracle separation with time-bounded measure would imply $\BPP$ is a measure 0 subset of $\EXP$.

	Our framework builds a foundation for studying permutation-based complexity using resource-bounded measure, in direct analogy to classical work on random oracles. It raises natural questions about the power and limitations of random permutations, their relationship to random oracles, and whether individual randomness can yield new class separations.

\end{abstract}

\newpage
\tableofcontents
\newpage

\section{Introduction}\label{sec:introduction}

The seminal work of Bennett and Gill \cite{BennettGill81}  established two foundational separations in computational complexity theory:
\begin{enumerate}
	\item $\P^A \neq \NP^A$ relative to a random oracle $A$ with probability 1.
	\item $\P^\pi \neq \NP^\pi \cap \coNP^\pi$ relative to a random permutation $\pi$ with probability 1.
\end{enumerate}
Subsequent research extended the first separation to hold for specific, individually random oracles, including algorithmically (Martin-L{\"o}f) random oracles \cite{Lutz:BLW}, polynomial-space-bounded random oracles \cite{Lutz:CSRPO}, and polynomial-time betting-game random oracles \cite{Hitchcock:PTROSCC}. However, the second separation has not yet been strengthened in an analogous way.
Whether $\P^A \neq \NP^A \cap \coNP^A$
holds relative to a random oracle $A$ remains an open question.

In this paper, we develop a novel framework for resource-bounded permutation measure and randomness, introducing {\em permutation martingales} and {\em permutation betting} games. These concepts generalize classical martingales and betting games to the space $\bfPi$ of all length-preserving permutations $\pi : \{0,1\}^* \to \{0,1\}^*$ where $|\pi(x)|=|x|$ for all $x \in \{0,1\}^*$.

\subsection{Background}

Bennett and Gill \cite{BennettGill81} initiated the study of random oracles in computational complexity, proving that $\P^A \neq \NP^A$ for a random oracle $A$ with probability 1. Subsequent work extended this to individual random oracles. Book, Lutz, and Wagner \cite{Lutz:BLW} showed that $\P^A \neq \NP^A$ for every oracle $A$ that is algorithmically random in the sense of Martin-L{\"o}f \cite{MartinLof66}. Lutz and Schmidt \cite{Lutz:CSRPO} improved this further to show $\P^A \neq \NP^A$ for every oracle $A$ that is pspace-random in the sense of resource-bounded measure \cite{Lutz:AEHNC}. Hitchcock, Sekoni, and Shafei \cite{Hitchcock:PTROSCC} extended this result to polynomial-time betting-game random oracles \cite{BvMRSS01}.

The complexity class \(\NP \cap \coNP\) is particularly significant because it comprises problems that have both efficiently verifiable proofs of membership and non-membership. This class includes important problems such as integer factorization and discrete logarithm, which are widely believed to be outside \(\P\) but are not known to be \(\NP\)-complete. These problems play a central role in cryptography, as the security of widely-used cryptosystems relies on their presumed intractability \cite{RSA78,DiffieHellman76}.
Furthermore, under derandomization hypotheses, \(\NP \cap \coNP\) has been shown to contain problems such as graph isomorphism \cite{KlivMe02}, further underscoring its importance in complexity theory. Thus, understanding the relationship between \(\P\) and \(\NP \cap \coNP\) relative to different notions of randomness could shed light on the structure of these classes and the limits of efficient computation.

\subsection{Our Approach: Permutation Martingales and Permutation Measure}

In this work, we develop a novel framework for resource-bounded permutation measure and randomness. We introduce permutation martingales and permutation betting games, extending classical notions of random permutations. Our theory captures essential properties of random permutations while enabling complexity separations. We prove that random oracles can be computed in polynomial time from a random permutation; however, the converse remains unresolved.

First, we recall the basics of resource-bounded measure.
A martingale in Cantor space may be viewed as betting on the membership of strings in a language. The standard enumeration of $\{0,1\}^*$ is $s_0 = \lambda, s_1 = 0, s_2 = 1, s_3 = 00, s_4 = 01, \ldots$. In the $\ith$ stage of the game, the martingale has seen the membership of the first $i$ strings and bets on the membership of $s_i$ in the language. The martingale's value is updated based on the outcome of the bet.
Formally, a classical martingale is a function \(d:\{0,1\}^* \to [0,\infty)\) satisfying the fairness condition
\[
	d(w) = \frac{d(w0) + d(w1)}{2}
\]
for all strings \(w\). Intuitively, \(d(w)\) represents the capital that a gambler has after betting on the sequence of bits in \(w\) according to a particular strategy. The fairness condition ensures that the expected capital after the next bit is equal to the current capital. A martingale succeeds on a language \(A \subseteq \{0,1\}^*\) if \[\limsupn d(A\restr n) = \infty,\] where \(A\restr n\) is the length-\(n\) prefix of \(A\)'s characteristic sequence. The {\em success set} of $d$ is $S^\infty[d]$, the set of all sequences that $d$ succeeds on. Ville \cite{Vill39} proved that a set $X$ has Lebesgue measure zero if and only if there is a martingale that succeeds on all elements of $X$. Lutz \cite{Lutz:AEHNC} defined resource-bounded measure by imposing computability and complexity constraints on the martingales in Ville's theorem.

We take a similar approach in developing resource-bounded permutation measure. Unlike a classical martingale betting on the bits of a language's characteristic sequence, a permutation martingale bets on the function values of a permutation $\pi$.
Instead of seeing the characteristic string of a language, a permutation martingale sees a list of permutation function values. More precisely, after $i \geq 0$ rounds of betting, a permutation martingale has seen a {\em prefix partial permutation}
\[ g = [g(s_0), \ldots, g(s_{i-1})] \] where $|g(s_i)|=|s_i|$ for all $i$. The permutation martingale will bet on the next function value $g(s_i)$. The current {\em betting length} is $l(g) = |s_i|$, the length of the next string $s_i$ in the standard enumeration.
The set of {\em free strings} available for the next function value is
\[ \free(g) = \myset{ x \in \{0,1\}^{l(g)} }{ x \text{ is not listed in } g  }. \]
For any prefix partial permutation $g$,
a permutation martingale $d$ outputs a value $d(g,x) \geq 0$ for each $x \in \free(g)$. The values satisfy the averaging condition
\[ d(g) = \frac{1}{|\text{free}(g)|}\sum_{x \in \text{free}(g)} {d(g, x)}. \]
Here $g,x$ denotes appending the string $x$ as the next function value in prefix partial permutation $g$. See Figure \ref{fig:permutation_martingale} for an example of a permutation martingale betting on strings of length 2.
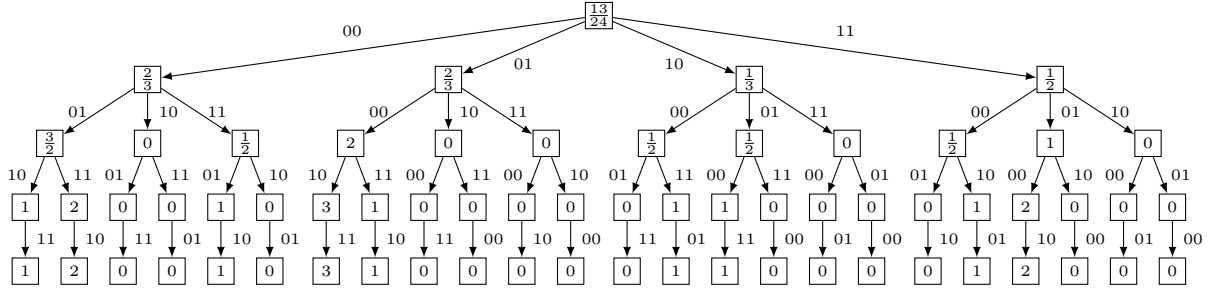
\begin{figure}
	\centering
	\begin{tikzpicture}[scale=0.5,
			level distance=1.7cm,
			level 1/.style={sibling distance=8cm},
			level 2/.style={sibling distance=2.6cm},
			level 3/.style={sibling distance=1.3cm},
			level_node/.style={rectangle, draw, minimum size=0.35cm, inner sep=0.5pt, fill=white,font=\tiny},leaf_node/.style={rectangle, draw, minimum size=0.35cm, inner sep=0.5pt,fill=white,font=\tiny},edge from parent/.style={draw, -latex},
			edge_label/.style={midway, font=\tiny,
					inner sep=4pt}
		]

		\node[level_node] {$\frac{13}{24}$}
		child { node[level_node] {$\frac{2}{3}$}
				child { node[level_node] {$\frac{3}{2}$}
						child { node[level_node] {$1$}
								child { node[leaf_node] {$1$} edge from parent node[edge_label, right] {11} }
								edge from parent node[edge_label, left] {10}
							}
						child { node[level_node] {$2$}
								child { node[leaf_node] {$2$} edge from parent node[edge_label, right] {10} }
								edge from parent node[edge_label, right] {11}
							}
						edge from parent node[edge_label, left] {01}
					}
				child { node[level_node] {$0$}
						child { node[level_node] {$0$}
								child { node[leaf_node] {$0$} edge from parent node[edge_label, right] {11} }
								edge from parent node[edge_label, left] {01}
							}
						child { node[level_node] {$0$}
								child { node[leaf_node] {$0$} edge from parent node[edge_label, right] {01} }
								edge from parent node[edge_label, right] {11}
							}
						edge from parent node[edge_label, right] {10}
					}
				child { node[level_node] {$\frac{1}{2}$}
						child { node[level_node] {$1$}
								child { node[leaf_node] {$1$} edge from parent node[edge_label, right] {10} }
								edge from parent node[edge_label, left] {01}
							}
						child { node[level_node] {$0$}
								child { node[leaf_node] {$0$} edge from parent node[edge_label, right] {01} }
								edge from parent node[edge_label, right] {10}
							}
						edge from parent node[edge_label, right] {11}
					}
				edge from parent node[edge_label, above left] {00}
			}
		child { node[level_node] {$\frac{2}{3}$}
				child { node[level_node] {$2$}
						child { node[level_node] {$3$}
								child { node[leaf_node] {$3$} edge from parent node[edge_label, right] {11} }
								edge from parent node[edge_label, left] {10}
							}
						child { node[level_node] {$1$}
								child { node[leaf_node] {$1$} edge from parent node[edge_label, right] {10} }
								edge from parent node[edge_label, right] {11}
							}
						edge from parent node[edge_label, left] {00}
					}
				child { node[level_node] {$0$}
						child { node[level_node] {$0$}
								child { node[leaf_node] {$0$} edge from parent node[edge_label, right] {11} }
								edge from parent node[edge_label, left] {00}
							}
						child { node[level_node] {$0$}
								child { node[leaf_node] {$0$} edge from parent node[edge_label, right] {00} }
								edge from parent node[edge_label, right] {11}
							}
						edge from parent node[edge_label, right] {10}
					}
				child { node[level_node] {$0$}
						child { node[level_node] {$0$}
								child { node[leaf_node] {$0$} edge from parent node[edge_label, right] {10} }
								edge from parent node[edge_label, left] {00}
							}
						child { node[level_node] {$0$}
								child { node[leaf_node] {$0$} edge from parent node[edge_label, right] {00} }
								edge from parent node[edge_label, right] {10}
							}
						edge from parent node[edge_label, right] {11}
					}
				edge from parent node[edge_label, below] {01}
			}
		child { node[level_node] {$\frac{1}{3}$}
				child { node[level_node] {$\frac{1}{2}$}
						child { node[level_node] {$0$}
								child { node[leaf_node] {$0$} edge from parent node[edge_label, right] {11} }
								edge from parent node[edge_label, left] {01}
							}
						child { node[level_node] {$1$}
								child { node[leaf_node] {$1$} edge from parent node[edge_label, right] {01} }
								edge from parent node[edge_label, right] {11}
							}
						edge from parent node[edge_label, left] {00}
					}
				child { node[level_node] {$\frac{1}{2}$}
						child { node[level_node] {$1$}
								child { node[leaf_node] {$1$} edge from parent node[edge_label, right] {11} }
								edge from parent node[edge_label, left] {00}
							}
						child { node[level_node] {$0$}
								child { node[leaf_node] {$0$} edge from parent node[edge_label, right] {00} }
								edge from parent node[edge_label, right] {11}
							}
						edge from parent node[edge_label, right] {01}
					}
				child { node[level_node] {$0$}
						child { node[level_node] {$0$}
								child { node[leaf_node] {$0$} edge from parent node[edge_label, right] {01} }
								edge from parent node[edge_label, left] {00}
							}
						child { node[level_node] {$0$}
								child { node[leaf_node] {$0$} edge from parent node[edge_label, right] {00} }
								edge from parent node[edge_label, right] {01}
							}
						edge from parent node[edge_label, right] {11}
					}
				edge from parent node[edge_label, below] {10}
			}
		child { node[level_node] {$\frac{1}{2}$}
				child { node[level_node] {$\frac{1}{2}$}
						child { node[level_node] {$0$}
								child { node[leaf_node] {$0$} edge from parent node[edge_label, right] {10} }
								edge from parent node[edge_label, left] {01}
							}
						child { node[level_node] {$1$}
								child { node[leaf_node] {$1$} edge from parent node[edge_label, right] {01} }
								edge from parent node[edge_label, right] {10}
							}
						edge from parent node[edge_label, left] {00}
					}
				child { node[level_node] {$1$}
						child { node[level_node] {$2$}
								child { node[leaf_node] {$2$} edge from parent node[edge_label, right] {10} }
								edge from parent node[edge_label, left] {00}
							}
						child { node[level_node] {$0$}
								child { node[leaf_node] {$0$} edge from parent node[edge_label, right] {00} }
								edge from parent node[edge_label, right] {10}
							}
						edge from parent node[edge_label, right] {01}
					}
				child { node[level_node] {$0$}
						child { node[level_node] {$0$}
								child { node[leaf_node] {$0$} edge from parent node[edge_label, right] {01} }
								edge from parent node[edge_label, left] {00}
							}
						child { node[level_node] {$0$}
								child { node[leaf_node] {$0$} edge from parent node[edge_label, right] {00} }
								edge from parent node[edge_label, right] {01}
							}
						edge from parent node[edge_label, right] {10}
					}
				edge from parent node[edge_label, above right] {11}
			};
	\end{tikzpicture}     \caption{An example permutation martingale on strings of length 2. Each path through the tree represents a permutation on $\{00,01,10,11\}$.}
	\label{fig:permutation_martingale}
\end{figure}

Prefix partial permutations may be used as cylinders to define a measure in $\bfPi$ that is equivalent to the natural product probability measure. We detail this in \Cref{sec:permutation_martingales_and_permutation_measure}. Briefly, a class $X \subseteq \bfPi$ has measure 0 if for every $\epsilon > 0$, there exists a sequence of cylinders $\{\cyl{g_i} \mid i \in \N\}$ that has total measure at most $\epsilon$ and covers $X$.
This is difficult to work with computationally as the covers may be large and require exponential time to enumerate.

We prove an analogue of Ville's theorem \cite{Vill39},
showing that permutation martingales characterize measure 0 sets in the permutation space $\bfPi$: a class $X \subseteq \bfPi$ has measure 0 if and only if there a permutation martingale $d$ with $X \subseteq S^\infty[d]$. This permutation martingale characterization allows us to impose computability and complexity constraints in the same way Lutz did for resource-bounded measure in Cantor space~\cite{Lutz:AEHNC}. In the following, let $\Delta$ be a resource bound such as $\p$, $\ptwo$, $\pspace$, or $\ptwospace$ (see \Cref{sec:resource-bounded_permutation_measure} for more details).
\begin{definition}
	Let $\Delta$ be a resource bound.
	A class of permutations $X \subseteq \bfPi$ has $\Delta$-measure 0 if there is a $\Delta$-computable permutation martingale that succeeds on $X$.
\end{definition}

Betting games \cite{BvMRSS01,MeMiNiStRe06} are a generalization of martingales that are allowed to bet on strings in an adaptive order rather than the standard order.
We analogously introduce permutation betting games as a generalization of both permutation martingales and classical betting games by allowing the betting strategy to adaptively choose the order in which it bets on the permutation's values. We use these betting games to define resource-bounded permutation betting-game measure.
\begin{definition}
	Let $\Delta$ be a resource bound.
	A class of permutations $X \subseteq \bfPi$ has $\Delta$-betting game measure 0 if there is a $\Delta$-computable permutation betting game that succeeds on $X$.
\end{definition}
We also define {\em individually} random permutations.
\begin{definition}
	Let $\pi \in \bfPi$ be a permutation and let $\Delta$ be a resource bound.
	\begin{enumerate}
		\item $\pi$ is $\Delta$-random if no $\Delta$-permutation martingale succeeds on $\pi$.
		\item $\pi$ is $\Delta$-betting game random if no $\Delta$-permutation betting game succeeds on $\pi$.
	\end{enumerate}
\end{definition}

\subsection{Our Results}

Our main result strengthens the Bennett--Gill permutation separation by proving that \(\P \neq \NP \cap \coNP\) relative to every polynomial-time betting-game random permutation \(\pi\). Formally, \Cref{th:main} establishes that
\[
	\P^\pi \neq \NP^\pi \cap \coNP^\pi
\]
for every \(\p\)-betting-game random permutation \(\pi\). In fact, we obtain even stronger separations in terms of bi-immunity \cite{FlajoletSteyaert74,BalSch85}, a notion formalizing the absence of infinite, easily-decidable subsets (see \Cref{sec:p_np_conp} for more details). We show that for a \(\p\)-betting-game random permutation \(\pi\), the class \(\NLIN^\pi \cap \co\NLIN^\pi\) contains languages that are bi-immune to \(\DTIME^\pi(2^{kn})\) for all \(k \geq 1\), where \(\NLIN\) denotes nondeterministic linear time. Moreover, relative to a \(\ptwo\)-betting-game random permutation, we derive that \(\NP^\pi \cap \coNP^\pi\) contains languages that are bi-immune to \(\DTIME^\pi(2^{n^k})\) for every \(k \geq 1\).

Bennett et al. \cite{BeBeBrVa97} showed
that  $\NP^\pi \cap \coNP^\pi \not\subseteq \BQTIME^\pi(o(2^{n/3}))$ relative to a random permutation $\pi$ with probability 1. We apply our resource-bounded permutation measure framework to improve this to individual space-bounded random oracles.
Specifically, we show that relative to a \(\ptwospace\)-random permutation \(\pi\),
\[
	\NP^\pi \cap \coNP^\pi \not\subseteq \BQP^\pi.
\]
This illustrates the power of our framework for analyzing the interplay between randomness, classical complexity, and quantum complexity.

\subsection{Random Oracles and Measure 0-1 Laws in \texorpdfstring{$\EXP$}{EXP}}

Tardos \cite{Tardos89} proved that  if $\AM \cap \coAM \neq \BPP$, then $\P^A \neq \NP^A \cap \coNP^A$ with probability 1 for a random oracle $A$.
\renewcommand{\ALMOST}[1]{\CCfont{ALMOST}\mbox{-}#1}
This is proved using $\CCfont{ALMOST}$ complexity classes.
For a relativizable complexity class $\calC$, its $\ALMOST{\calC}$ class consists of all languages that are in the class with probability 1 relative to a random oracle:
$\ALMOST{\calC} = \{ L \mid \Pr[ L \in \calC^A ] = 1 \}.$
We have $\ALMOST{\P} = \BPP$ \cite{BennettGill81} and $\ALMOST{\NP} = \AM$ \cite{NisWig94}. The condition $\AM \cap \coAM \neq \BPP$ implies that there exist problems in $\ALMOST{\NP}\cap\ALMOST{\coNP}$ that are not in $\ALMOST{\P}$. Since the intersection of measure 1 classes is measure 1, this implies $\NP^A \cap \coNP^A \neq \P^A$ relative to a random oracle $A$ with probability 1.
Recent work of Ghosal et al. \cite{Ghosal:STOC23} shows that if $\UP \not\subseteq \RP$, then $\P^A \neq \NP^A \cap \coNP^A$ with probability 1 for a random oracle $A$.
In \Cref{sec:limitations} we pivot from permutation randomness to classical random oracles and show that resolving the long‑standing question ``does $\P^R = \NP^R\cap\coNP^R$ with probability 1?'' is tightly linked to quantitative structure inside $\EXP$.  Leveraging the conditional oracle separations of Tardos \cite{Tardos89} and of Ghosal et al. \cite{Ghosal:STOC23}, we prove that if $\P^R=\NP^R\cap\coNP^R$ holds almost surely, then several familiar subclasses of $\EXP$ obey strong 0-1 laws: specifically, either $\NPcoNP$, $\UPcoUP$,  (and, in a weaker form, $\UP$ vs. $\FewP$) each has $\p$-measure 0 or else fills all of $\EXP$.  Consequently, non‑measurability of any one of these classes immediately forces $\P^R\neq \NP^R\cap\coNP^R$ with probability 1.  We further show that placing the same oracle separation in $\ptwo$ measure would collapse BPP below EXP, thereby framing the random-oracle problem in terms of concrete measure-theoretic thresholds inside exponential time.

\subsection{Organization}

This paper is organized as follows: \Cref{sec:preliminaries} contains preliminaries. \Cref{sec:permutation_martingales_and_permutation_measure} develops permutation martingales, resource-bounded permutation measure, and random permutations. Elementary properties of $
	\p$-random permutations are presented in \Cref{sec:p_rand_permutations_app}.
In \Cref{sec:p_np_conp}, we prove our main results on random permutations for $
	\P$ vs. $\NP \cap \coNP$. \Cref{sec:quantum} contains our results on $\NPcoNP$ versus quantum computation relative to a random permutation. In \Cref{sec:limitations} we present our results on random oracles and 0-1 laws.
We conclude in \Cref{sec:conclusion} with some open questions.

\section{Preliminaries}\label{sec:preliminaries}
The binary alphabet is $\Sigma = \{0,1\}$,
the set of all binary strings is $\Sigma^*$, the set of all binary strings of
length $n$ is $\Sigma^n$, and the set of all infinite binary sequences is
$\Sigma^\infty$. The empty string is denoted by $\lambda$.  We use
the standard enumeration of strings,
$s_0=\lambda,s_1=0,s_2=1,s_3=00,s_4=01,\ldots$.
The characteristic sequence of a language $A$ is the sequence
$\chi_A \in \Sigma^\infty$, where
$\chi_A[n] = 1 \iff s_n \in A$. We refer to $\chi_A[s_n] = \chi_A[n]$ as the
characteristic bit of $s_n$ in $A$. A language $A$ can alternatively be seen
as a subset of $\Sigma^*$, or as an element of $\Sigma^\infty$ via
identification with its characteristic sequence $\chi_A$. Given strings
$x,y$
we denote by $[x,y]$ the set of all strings $z$ such that $x\leq z \leq y$.
For any string $s_n$ and natural number $k$, $s_n+k$ is the string $s_{n+k}$; e.g.
$\lambda + 4 = 01$. Similarly we denote by $A[x,y]$ the substring of the
characteristic sequence $\chi_A$ that corresponds to the characteristic bits
of the strings in $[x,y]$. We use parentheses for intervals that do not
include the endpoints. We write $A\restr n$ for the length $n$ prefix of $A$.  A statement $\mathcal{S}_n$ holds infinitely
often (written i.o.) if it holds for infinitely many $n$, and it holds almost
everywhere (written a.e.) if it holds for all but finitely many $n$.

\section{Permutation Martingales and Permutation Measure}\label{sec:permutation_martingales_and_permutation_measure}
\subsection{Permutation Measure Space}

Resource-bounded measure is typically defined in the Cantor Space
$\C = \{0,1\}^\infty = 2^\N$ of all infinite binary sequences.
For measure in $\C$, we use the open balls or cylinders $\C_w = w \cdot \C$ that
have measure $\mu(\C_w) = 2^{-|w|}$ for each $w \in \Sigma^*$. Let
$\calC$ be the $\sigma$-algebra generated by
$\{ \C_w \mid w \in \{0,1\}^* \} $. Resource-bounded measure and algorithmic randomness typically work in the probability space $(\C,\calC,\mu)$.

We only consider permutations in $\bfPi$, the set of permutations on
$\binary^*$ that preserve string lengths. Given a permutation
$\pi \in \bfPi$, we denote by $\pi_n$ the permutation $\pi$ restricted to
$\binaryn$ i.e., $\pi_n$ is a permutation on $\binaryn$. Similarly, $\bfPi_n$
denotes the set of permutations in $\bfPi$ restricted to $\binaryn$. Bennett and Gill \cite{BennettGill81} considered random permutations by placing the uniform measure on each $\bfPi_n$ and taking the product measure to get a measure on $\bfPi$. We now define this measure space more formally so we may place martingales on it.

Standard resource-bounded measure identifies a language $A \subseteq \{0,1\}^*$ with its infinite binary characteristic sequence $\chi_A \in \C$.
For permutations, we analogously use the value sequence consisting of all function values.
\begin{definition}
	The {\em value sequence} of a permutation $f \in \bfPi$ is the sequence
	\[\nu_f = [f(s_0),f(s_1),f(s_2),\ldots]\]
	of function values where $s_0,s_1,s_2,\ldots$ is the standard enumeration of $\{0,1\}^*$.
\end{definition}
We identify a permutation $f \in \bfPi$ with its value sequence $\nu_f$. Initial segments of permutations are called prefix partial permutations.

\begin{definition}
	A {\em prefix partial permutation} is a list
	$g = [g(s_0),\ldots, g(s_{N-1})]$
	of function values for some $N \geq 0$ where no value is repeated and  $|g(s_i)| = |s_i|$ for all $0 \leq i < N$.
	We let $\PFPi$ denote the class of all
		{prefix partial permutations}.
\end{definition}

We write each $g \in \PFPi$ as a list $g = [g(s_0),\ldots,g(s_{N-1})]$. The {\em length} of $g$ is $N$, the number of function values assigned, and is denoted $|g|$.  We use $\emptylist$ to denote the {\em empty list}, the list of length 0. We write $f\restr N$ for the length $N$ prefix partial permutation of $f \in \bfPi$.

\begin{definition}
	For each $g = [g(s_0),\ldots,g(s_{N-1})] \in \PFPi$, the {\em cylinder} of all permutations in $\bfPi$ that extend $g$ is
	\[\cyl{g} = \{ h \in \bfPi \mid h(s_0)=g(s_0),\ldots,h(s_{N-1})=g(s_{N-1}) \}.\]
\end{definition}

For measure in ${\bfPi}$, we are taking the uniform distribution on
the set of all $\bfPi_n$ of length-preserving permutations for all $n$. Our basic open sets are $\{\cyl{g} \mid g \in \PFPi\}$.
Suppose $g \in \PFPi$ has $|g| = 2^{n}-1$ for some $n \geq 0$. Then, following Bennett and Gill \cite{BennettGill81}, the measure
\[\mu(\cyl{g}) = \prod_{k=0}^{n-1} \frac{1}{(2^k)!}\]
is  easy to define because the distribution is uniform over the $(2^k)!$ permutations at each length.
If $2^{n}-1 \leq |g| < 2^{n+1} -1$, let $m = |g| - 2^{n} +1$ and then
\[\mu(\cyl{g})
	= \left(\prod_{k=0}^{n-1} \frac{1}{(2^k)!} \right) \frac{(2^{n}-m)!}{(2^{n})!}
	= \left(\prod_{k=0}^{n-1} \frac{1}{(2^k)!} \right) \frac{1}{P(2^{n},m)},\]
where $P(n,k) = \frac{n!}{(n-k)!}$ denotes the number of $k$-permutations on $n$ elements. For convenience, we commonly write $\mu(g) = \mu(\cyl{g}).$

Let
$\calF_\bfPi = \sigma(\PFPi)$ be the $\sigma$-algebra generated by the collection of all $\cyl{g}$ where $g \in \PFPi.$
By Carath\'eodory's extension theorem, $\mu$ extends
uniquely to $\calF_\bfPi$, yielding the probability space
$(\bfPi,\calF_\bfPi,\mu).$
We will work in this probability space.
Because $\mu$ is outer regular, we have the typical open cover characterization of measure zero:

\begin{theorem}
	A class $X \subseteq \bfPi$ has {\em measure 0} if and only if for every $\epsilon > 0$, there is an open covering $G = \{ g_0, g_1, \ldots, \} \subseteq \PFPi$ such that
	\[\sum\limits_{i=0}^\infty \mu(g_i) < \epsilon
		\quad\text{and}\quad
		X \subseteq \bigcup\limits_{i=0}^\infty \cyl{g_i}.\]
\end{theorem}

\subsection{Permutation Martingales}

In resource-bounded measure in Cantor Space, a {\em martingale} is a function $d : \Sigma^* \to
	[0,\infty)$ such that for all $w \in \Sigma^*$, we have the
following averaging condition:
\[d(w) = \frac{d(w0)+d(w1)}{2}.\]
A martingale in Cantor space may be viewed as betting on the membership of strings in a language. The standard enumeration of $\{0,1\}^*$ is $s_0 = \lambda, s_1 = 0, s_2 = 1, s_3 = 00, s_4 = 01, \ldots$. In the $\ith$ stage of the game, the martingale has seen the membership of the first $i$ strings and bets on the membership of $s_i$ in the language. The martingale's value is updated based on the outcome of the bet. For further background on resource-bounded measure,
we refer to \cite{Lutz:AEHNC,Lutz:QSET,AmbMay97,BvMRSS01,Harkins:ELBG}.

A permutation martingale operates similarly, but instead of betting on the membership of a string in a language it bets on the next function value of the permutation. Instead of seeing the characteristic string of a language, a permutation martingale sees a {\em prefix partial permutation}, which is a list of permutation function values $g = [g(s_0), \ldots, g(s_{i-1})]$ satisfying $|g(s_i)| = |s_i|$ for all $i$.
The permutation martingale will bet on the next function value $g(s_i)$. The current {\em betting length} is the length of the next string $s_{|g|}$ in the standard enumeration:
$ l(g) = |s_{|g|}|. $
The set of {\em free strings} available for the next function value is
\[ \free(g) = \{ x \in \{0,1\}^{l(g)} \mid x \text{ is not in } g  \}. \]
For example,
$\free([\lambda]) = \{0,1\}$, $\free([\lambda, 1, 0, 11]) = \{00, 01, 10\}$,
and $\free([\lambda, 1, 0, 11, 00, 01]) = \{10\}$.

We now introduce our main conceptual contribution, permutation martingales.
\begin{definition}
	A {\em permutation martingale} is a function
	$d: \PFPi \to [0, \infty)$
	such that for every prefix partial permutation $g \in \PFPi$,
	\[d(g) = \frac{1}{|\free(g)|}\sum_{x\in\free(g)} d(g, x),\]
	where $(g, x)$ is the result of appending $x$ to $g$.
\end{definition}

Success is defined for permutation martingales analogously to success for classical martingales.

\begin{definition}
	Let $d$ be a permutation martingale.
	We say $d$ {\em succeeds on} $f \in \bfPi$ if \[\limsupN d(f \restr N) = \infty.\]
	The {\em success set} of $d$ is \[S^\infty[d] = \{ f \in \bfPi \mid d\textrm{ succeeds on } f\}\]
	and the {\em unitary success set} of $d$ is the set
	\[S^1[d] = \{ f \in \bfPi \mid (\exists n)\ d(f\restr n) \geq 1 \}.\]
\end{definition}

In the remainder of this section, we establish the analogue of Ville's theorem \cite{Vill39} for measure in $\bfPi$ and permutation martingales.

\begin{theorem}\label{th:ville_permutation_measure}
	The following statements are equivalent for every $X \subseteq \bfPi$:
	\begin{enumerate}
		\item $X$ has measure 0.
		\item For every $\epsilon > 0$, there is a permutation martingale $d$ with $d(\lambda) < \epsilon$ and $X \subseteq S^1[d]$.
		\item There is a permutation martingale $d$ with $X \subseteq S^\infty[d]$.
	\end{enumerate}
\end{theorem}

First, we need a few lemmas.

\begin{lemma}\label{le:permutation_martingale_cylinder}
	If $g \in \PFPi$, then there is a permutation martingale $d_g$ with $d_g(\lambda) = \mu(g)$ and $S^1(d_g) = \cyl{g}$.
\end{lemma}
\begin{proof}Let $N = |g|$ and define $d_g(x) = \PR{|h|=N}{g \mid x \prefix h}$, where we choose $h \in \PFPi$ of length $N$ uniformly at random.
\end{proof}

A {\em premeasure} on $\bfPi$ is a function $\rho: \PFPi \to [0,1]$ such that $\rho(\lambda) = 1$ and for all $g \in \PFPi$,
$$\rho(g) = \sum_{w \in \free(g)} \rho(g,w).$$
A {\em prefix set} in $\PFPi$ is a set $W \subseteq \PFPi$ such that no element of $W$ is a prefix of any other element.

\begin{lemma}\label{le:kraft_inequality_PPF}
	If $W \subseteq \PFPi$ is a prefix-free set and $\rho$ is any premeasure on $\bfPi$, then
	$$\sum_{g \in W} \rho(g) \leq 1.$$
\end{lemma}\begin{appendixproof}[Proof of \Cref{le:kraft_inequality_PPF}]
	Because $W$ provides a disjoint collection of cylinders, the total measure of the cylinders is at most $1$. (A rigorous proof may be given using induction.)
\end{appendixproof}

\begin{lemma}\label{le:kraft_inequality_martingales}
	If $W$ is a prefix-free set and $d$ is any permutation martingale, then $$\sum_{g\in W} \mu(g) d(g) \leq d(\emptylist).$$
\end{lemma}

\begin{proof}
	The function $\mu(g)d(g)$ is a premeasure on $\bfPi$, so this follows from \Cref{le:kraft_inequality_PPF}
\end{proof}

For $k \geq 1$, let $S^k[d] =\{ \pi \in \bfPi \mid (\exists n) d(\pi \restr n) \geq k \}$.
\begin{lemma}\label{le:permutation_martingale_conservation}
	For any permutation martingale and $k \geq 1$, $$\mu(S^k[d]) \leq \frac{d(\emptylist)}{k}.$$
	and there is an open cover $W_k$ with $\mu(W_k) = \mu(S^k[d])$.
\end{lemma}

\begin{proof}Let $d$ be a permutation martingale.
	Define for each $k \geq 1$, $$W_k = \{ g \in \PFPi \mid d(g) \geq k \textrm{ and }d(h) < k\textrm{ for all proper prefixes $h$ of $g$} \}.$$
	Then $$S^k[d] = \bigcup_{g \in W_k} \cyl{g},$$
	so $\mu(W_k) = \mu(S^k[d]).$
	By  \Cref{le:kraft_inequality_martingales},
	$$d(\emptylist) \geq \sum_{g \in W_k} \mu(g)d(g) \geq \sum_{g \in W_k} \mu(g) k = k \sum_{g\in W_k} \mu(g) = k \mu(W_k).$$
\end{proof}

We are now ready to prove \Cref{th:ville_permutation_measure}.

\begin{proof}[Proof of \Cref{th:ville_permutation_measure}]
	Suppose 1 is true. Let $G$ be a covering of $X$ with $\mu(G) < \epsilon$.
	Then define $d$ by $d = \sum\limits_{g \in G} d_g$ where each $d_g$ comes from \Cref{le:permutation_martingale_cylinder}. We have $d(\lambda) = \mu(G) < \epsilon$ and $X \subseteq S^1[d]$.

	Suppose 2 is true. For each $k \in \N$, let $d_k$ be a martingale with $d_k(\lambda) < 2^{-k}$ and $X \subseteq S^1[d_k]$. Without loss of generality, we assume that if $d_k(g) \geq 1$, then $d_k(h) = d_k(g)$ for all $g \prefix h$.
	Define $d$ by $d = \sum\limits_{k \in \N} d_k$. Let $A \in X$. For every $k$, there exists $n_k$ such that $d_k(A \restr n_k) \geq 1$. Let $m_k = \max\{n_1, \ldots, n_k\}$. Then $d(A \restr m_k) \geq k$. Since $k$ is arbitrary, $A \in S^\infty[d]$. Therefore $X \subseteq S^\infty[d]$.

	Suppose 3 is true. Let $d$ be a martingale with $d(\lambda) = 1$ and $X \subseteq S^\infty[d]$. Let $\epsilon > 0$ and let $k \geq 1$ such that $\frac{1}{k} < \epsilon$. Let $B_k$ be the set of all shortest $g$ with $d(g) \geq k$. Then $X \subseteq B_k$, $B_k$ is an open set, and $\mu(B_k) \leq \frac{1}{k} < \epsilon$ by \Cref{le:permutation_martingale_conservation}.
\end{proof}

\subsection{A Permutation Martingale Example}

We construct a permutation martingale \(d\) that succeeds on any length-preserving permutation whose restriction to length~\(n\) is a cycle permutation for all but finitely many \(n\).
We partition the initial capital into infinitely many shares \(a_i = 1/i^2\).
For each \(i\), the share \(a_i\) is used to bet on the event that, for all \(n \ge i\), the length-\(n\) restriction of the permutation is a cycle permutation.

The betting strategy is simple: when moving from length \(n-1\) to \(n\), the martingale wagers all relevant capital on the image of the \(n\)-bit string \(1^{n-1}0\). In the final step of forming a cycle of length \(2^n\), there are exactly two choices for the image of \(1^{n-1}0\).
One choice yields a cycle of length \(2^n\); the other does not.
Since it is a binary choice, the martingale places its entire stake \(a_i\) (for all \(i \le n\)) on the cycle outcome, thereby doubling its capital whenever the cycle is formed.

Hence, on any permutation whose restriction to length~\(n\) is a cycle permutation for all but finitely many \(n\), infinitely many of these bets succeed.
Consequently, each of those corresponding shares \(a_i\) grows without bound, and so the overall martingale \(d\) succeeds on all such permutations.

\medskip

\noindent

\medskip

\noindent

\subsection{Permutation Martingales as Random Variables}

Hitchcock and Lutz \cite{Hitchcock:WCCRSM} showed how the martingales used in computational complexity are a special case of martingales used in probability theory. We explain how this extends to permutation martingales. Given a martingale $d : \{0,1\}^* \to [0,\infty)$, Hitchcock and Lutz define the random variable $\xi_{d,n} : \C \to [0,\infty)$ by
$\xi_{d,n}(S) = d(S \restr n)$
for each $n \geq 0$.
Let $\calM_n = \sigma( \{ \C_w \mid w \in \{0,1\}^n \} )$ be the $\sigma$-algebra generated by the cylinders of length $n$.
Then the sequence $(\xi_{d,n} \mid n \geq 0)$ is a martingale in the probability theory sense with respect to the filtration $(\calM_n \mid n \geq 0)$: for all $n \geq 0$,
$E[\xi_{d,n+1} \mid \calM_n] = \xi_{d,n}.$

Similarly, given a permutation martingale $d : \PFPi \to [0,\infty)$, for each $N$ we can define the random variable
$X_{d,N} : \bfPi \to [0,\infty)$ by
$X_{d,N}(f) = d(f \restr N)$ for each $ N \geq 0$.
Let \[\calG_N = \sigma( \{ \cyl{g} \mid g \in \PFPi\textrm{ and }|g| = N \})\] be the $\sigma$-algebra generated by the  cylinders in $\PFPi$ of length $N$.
Then $(X_{d,N}\mid N \geq 0)$ is a martingale in the probability theory sense with respect to the filtration $(\calG_N \mid N \geq 0)$: for all $N \geq 0$,
$E[X_{d,N+1} \mid \calG_N] = X_{d,N}.$

\subsection{Resource-Bounded Permutation Measure}\label{sec:resource-bounded_permutation_measure}

We follow the standard notion of computability for real-valued functions \cite{Lutz:AEHNC} to define resource-bounded permutation martingales.
\begin{definition}\label{def:resource_bounded_permutation_martingales}
	Let $d : \PFPi \to [0,\infty)$ be a permutation martingale.
	\begin{enumerateC}
		\item   $d$ is {\em computable in time $t(n)$} if there  is an exactly computable $\hat{d} : \PFPi \times \N \to \Q$ such
		that for all $f \in \PFPi$ and $r \in \N$, $|d(f) - \hat{d}(f,r)| \leq 2^{-r}$ and $\hat{d}(f,r)$ is
		computable in time $t(|f|+r)$.
		\item   $d$ is {\em computable in space $s(n)$} if there  is an exactly computable $\hat{d} : \PFPi \times \N \to \Q$ such
		that for all $f \in \PFPi$ and $r \in \N$, $|d(f) - \hat{d}(f,r)| \leq 2^{-r}$ and $\hat{d}(f,r)$ is
		computable in space $s(|f|+r)$.
		\item If $d$ is {computable in polynomial time}, then $d$ is a {\em $\p$-permutation martingale}.
		\item If $d$ is {computable in quasipolynomial time}, then $d$ is a {\em $\ptwo$-permutation martingale}.
		\item If $d$ is {computable in polynomial space}, then $d$ is a {\em $\pspace$-permutation martingale}.
		\item If $d$ is {computable in quasipolynomial space}, then $d$ is a {\em $\ptwospace$-permutation martingale}.
	\end{enumerateC}
\end{definition}

We are now ready to define resource-bounded permutation measure.

\begin{definition} Let $\Delta \in \{\p,\ptwo,\pspace,\ptwospace\}$. Let  $X \subseteq \bfPi$ and $X^c = \bfPi - X$ be the complement of $X$ within $\bfPi$.
	\begin{enumerateC}
		\item $X$ has {\em $\Delta$-measure 0}, written $\mu_\Delta(X) = 0$, if there is a $\Delta$-computable permutation martingale $d$ with $X \subseteq S^\infty[d]$.
		\item $X$ has {\em $\Delta$-measure 1}, written $\mu_\Delta(X) = 1$, if $\mu_\Delta(X^c) = 0$
	\end{enumerateC}
\end{definition}

\begin{definition} Let $\Delta \in \{\p,\ptwo,\pspace,\ptwospace\}$.
	A permutation $\pi \in \bfPi$ is {\em $\Delta$-random} if $\pi$ is not contained in any $\Delta$-measure 0 set.
\end{definition}
Equivalently, $\pi$ is $\Delta$-random if no $\Delta$-martingale succeeds on $\pi$.

\subsection{Permutation Betting Games}

Originated in the field of algorithmic information theory, betting games are a generalization of martingales \cite{Muchnik1998,MeMiNiStRe06}, which were introduced to computational complexity by Buhrman et al.~\cite{BvMRSS01}. Similar to martingales, betting games can be thought of as strategies for betting on a binary sequence, except that with betting games we have the additional capability of selecting which position in a sequence to bet on next. In other words, a betting
game is permitted to select strings in a nonmonotone order, with the important restriction that it may not bet on
the same string more than once (see Buhrman et al. \cite{BvMRSS01} for more details).

A permutation betting game is a
generalization of a permutation martingale, implemented by an oracle Turing machine, where it is allowed to select strings
in nonmonotone order. Prefixes of permutation betting games can be represented as {\em ordered partial permutations} defined below.
\begin{definition}
	An {\em ordered partial permutation} is a list
	$g = [(x_1,y_1),\ldots, (x_n, y_n)]$
	of pairs of strings for some $n \geq 0$ where for all $1 \leq i < j \leq n$, $x_i \neq x_j$ and $y_i \neq y_j$, and   $|x_i| = |y_i|$ for all $1 \leq i \leq n$.
	We let $\OPPi$ denote the class of all
		{ordered partial permutations}.
\end{definition}
For a permutation betting game, the averaging
condition takes into consideration the length of the next string to be queried as follows.
Let $w \in \OPPi$ be the list of queried strings
paired with their images, and $a \in \binaryn$ be the next string the betting game will query. Define free$(w,n)$ to be the set of length-$n$ strings that are available for the next function value, i.e., length-$n$ strings that are not the function value of any of the queried strings.
Then the following averaging condition over free strings of length $n$ must hold for the permutation betting game $d : \OPPi \to [0,\infty) $

\[
	d(w) = \sum_{b\in\free(w,n)} \frac{d(w[a, b])}{|\free(w,n)|}\]
where $w[a,b]$ is the list $w$ appended with the pair $(a, b)$.

\begin{definition}
	A betting game is a {\em $t(n)$-time betting game} if for all $n$, all strings of length $n$ have been queried by time $t(2^n)$.
\end{definition}

We define betting game measure 0 and betting game randomness analogously.

\begin{definition} Let $\Delta \in \{\p,\ptwo,\pspace,\ptwospace\}$. \begin{enumerate}
		\item A class $X\subseteq \bfPi$ has {\em $\Delta$-betting-game measure 0} if there is a $\Delta$-computable permutation betting game $d$ with $X \subseteq S^\infty[d]$.
		\item    A permutation $\pi \in \bfPi$ is {\em $\Delta$-betting game random} if no $\Delta$-betting game succeeds on $\pi$.
	\end{enumerate}
\end{definition}

{\subsection{Measure Conservation}

Lutz's Measure Conservation Theorem implies that resource-bounded measure gives nontrivial notions of measure within exponential-time complexity classes: $\mup(\E) \neq 0$ and $\muptwo(\EXP) \neq 0$.
Let $\permE$ be the class of length-preserving permutations that can be
computed in $2^{O(n)}$ time and $\permEXP$ be the class of length-preserving
permutations that can be computed in $2^{n^{O(1)}}$ time.
We show that our notions of permutation measure have conservation theorems within these classes of exponential-time computable permutations.

\begin{lemma}
	\label{lem:perm_p_i}
	For any $t(2^n)$-time permutation martingale $D$, we can construct a $O(2^{2n}t(2^n + 2^{2n}))$-time permutation that is not succeeded on by $D$.
\end{lemma}
\begin{appendixproof}[Proof of \Cref{lem:perm_p_i}]
	Let $Q$ be a Turing machine that operates as follows on input $x \in \{0,1\}^n$.
	The machine $Q$ simulates the permutation martingale $D$ on prefix partial
	permutations starting from the empty list and ending when the image of $x$ is added
	to the list. During the simulation, whenever $D$ bets on a string $y$, $Q$ maps $y$
	to the first string to which no other string has been mapped, and for which $D$’s
	capital decreases the most, up to an additive error of $2^{-2n}$. After this
	simulation ends, $D$ outputs the string assigned to $x$ by this simulation. Note that when computing $Q(x)$, we first compute $Q(y)$ for all $y<x$.

	Clearly, $Q$ computes a permutation on which $D$’s capital never exceeds
	its initial value plus \[\sum_{n=1}^\infty\sum_{x\in\{0,1\}^n}2^{-2n} = \sum_{n=1}^\infty 2^{-n} = 1,\] so $D$ cannot succeed
	on this permutation. The runtime of $Q$ on input $x$ is
	$O(2^{2n}t(2^n + 2^{2n}))$, since finding the image that minimizes $D$’s capital on a length-$n$ string $x$ requires scanning $O(2^n)$ candidates for each of the $O(2^n)$ strings preceding $x$, and each evaluation of the martingale’s value (to within an additive error of $2^{-2n}$)
	takes $t(2^n + 2^{2n})$ time.
\end{appendixproof}

The following theorem follows from \Cref{lem:perm_p_i}.
\begin{theorem}\label{thm:measure_conservation_1}
	\begin{enumerate}
		\item  $\permE$ does not have $\p$-permutation measure 0.
		\item  $\permEXP$ does not have $\p_2$-permutation measure 0.
	\end{enumerate}
\end{theorem}

Proving similar results for betting games turns out to be more challenging, given that they are allowed to bet on strings in an adaptive order. To address this, we define the following class of {\em honest} betting games.

\begin{definition}
	A {\em $\log(t(2^n))$-honest $t(n)$-permutation betting game} is a $t(n)$-time permutation betting game
	such that for all languages $A$, for all $n$, all non-zero bets by time $t(2^n)$ are for strings of length at most $\log(t(2^n))$.

\end{definition}

We use this definition in the following lemma.
\begin{lemma}
	\label{lem:perm_bg_p_i}
	For any $\log(t(2^n))$-honest $t(2^n)$-permutation betting game $G$, we can
	construct a $O\bigr(t(2^n)^2\, t(t(2^n) + t(2^n)^2)\bigl)$-time permutation that is not succeeded on by $G$.
\end{lemma}

\begin{appendixproof}[Proof of \Cref{lem:perm_bg_p_i}]
	This proof is similar to the proof of \Cref{lem:perm_p_i}. So we
	focus on the main difference between the proofs.

	We construct a $\DTIME\bigl(t(2^n)^2\, t(t(2^n) + t(2^n)^2)\bigr)$ permutation using the same
	strategy as was used in the previous proof. The only difference is when the simulated betting game $G$  queries a string of length greater than
	$m = \log(t(2^n))$. In this case, the betting game answers the query with the
	first string of length $m$ that has not yet been made an image. This step
	takes $O(t(2^n))$ time. Since the betting game does not bet on this string,
	its capital is unaffected. When the betting game queries strings of length at
	most $m$, it behaves just like it did in the previous proof, it finds a string
	that minimizes the betting game's capital within an additive factor of
	$2^{-2n}$.

	It is easy to see that the betting game does not succeed on this
	permutation. This permutation can be computed in \(O\bigl(t(2^n)^2\, t(t(2^n) + t(2^n)^2)\bigr)\). This is because on input $x \in \{0,1\}^n$, the simulation queries at most $t(2^n)$ strings, for each queried string it examines at most $t(2^n)$ strings to be its image, and it takes $O(t(2^{\log(t(2^n))} + 2^{2\log(t(2^n)}) = O(t(t(2^n) + t(2^n)^2)$ time to evaluate the betting on each image within the appropriate margin of error.
\end{appendixproof}

\begin{theorem} \label{thm:measure_conservation_2}
	\begin{enumerate}
		\item $\permE$ does not have $O(n)$-honest $\p$-permutation betting game measure 0.
		\item  $\permEXP$ does not have  $O(n^k)$-honest $\p_2$-permutation betting game measure 0.
	\end{enumerate}
\end{theorem}

Since $\pspace$-permutation martingales can simulate $O(n)$-honest $\p$-betting-games, and $\ptwospace$-permutation martingales can simulate $O(n^k)$-honest $\ptwo$-betting-games, we have the following:

\begin{proposition}\label{prop:pspace_random_implies_honest_p_betting_game_random}
	Let $\pi$ be a permutation.
	\begin{enumerate}
		\item If $\pi$ is a $\pspace$-random permutation, then $\pi$ is $O(n)$-honest $\p$-betting game random.
		\item If $\pi$ is a $\ptwospace$-random permutation, then $\pi$ is $O(n^k)$-honest  $\ptwo$-betting game random.
	\end{enumerate}
\end{proposition}
}

\section{Elementary Properties of Random Permutations}\label{sec:p_rand_permutations_app}

In this section, we explore fundamental properties of random permutations that provide insights into how permutation martingales and betting games operate. Understanding these properties is crucial for applying permutation randomness in computational complexity.
We show that random permutations are computationally difficult to  compute and to invert. We then investigate the relationship between random permutations and random oracles, showing how random permutations can generate random oracles.

\subsection{Intractability of Random Permutations}

\begin{definition}
	A permutation $\pi \in \bfPi$ is {\em noticeably polynomial time} if there are polynomials $p,q$ and $\TM$ $M$ such that for infinitely many $n$, $M$ computes $\pi$ on at least $2^n/p(n)$ strings of length $n$ in $q(n)$ time for each string.
\end{definition}

\begin{theorem}\label{thm:NoticeablePermutationMeasure}
	The set
	$X = \{ \pi \in \bfPi \mid \pi_n \textrm{ is noticeably polynomial time} \}$
	has $\p$-permutation measure 0.
\end{theorem}
\noindent

The proof uses a simple averaging argument: we partition the set of length-\(n\) strings into \(2^n / n^{\lg n}\) subintervals, each of size \(n^{\lg n}\). By the noticeably-polynomial-time property of the permutations in \(X\), at least one subinterval contains superpolynomially many strings whose images are computable. The martingale then identifies a sufficiently small subset of these strings and makes correct predictions often enough to succeed.

\begin{proof}[Proof of \Cref{thm:NoticeablePermutationMeasure}]
	We design a polynomial-time permutation-martingale game that succeeds
	on $X$. The martingale succeeds on $X$ by going through all polynomial time
	TMs and finding a size-$n^{1/2}$ subset of $\binaryn$ where it wins enough
	bets to succeed.

	This martingale operates in stages. In stage $n$, the martingale uses the
	first $n$ polynomial TMs to bet on $\binaryn$. We partition $\binaryn$ into
	$2/n^{\lg n}$ subsets of size $n^{\lg n}$. For each partition, we consider
	all its size-$n^{1/2}$ subsets. Then we run the $n$ TMs on each
	size-$n^{1/2}$ subset of each partition. So, in total we consider
	$N_n = (2^n/n^{\lg n})  \binom{n^{\lg n}}{n^{1/2}}=2^{\Theta(n)}$ subsets.
	We use the TMs outputs to predict the permutation we are betting on. Each
	size-$n^{1/2}$ subset and TM is given capital $\frac{1}{n^2N_n}$ to make
	its bet. Next, we show that for any $\pi \in X$, for infinitely many $n$,
	there is a subset that increases its capital to $\omega(1)$.

	We now specify how the martingale bets in stage $n$. First, we naturally
	partition $\binaryn$ into $2^n/n^{\lg n}$ contiguous subsets
	$T_i = [0^n + in^{\lg n}, 0^n + (i+1)n^{\lg n})$, for
	$i \in [0, 2^n/n^{\lg n})$. For each $T_i$, the martingale examines all its
	size-$n^{1/2}$ subsets. It then runs each size-$n^{1/2}$ subset on the first
	$n$ TMs, $M_k$ for $k \in [1,n]$. Let $M_{i,j,k}$ denote the image of the
	$j$th size-$n^{1/2}$ subset of $T_i$ under the function computed by $M_k$.
	As mentioned before, we use $\frac{1}{n^2N_n}$ to bet with each $M_{i,j,k}$.
	The martingale uses $\frac{1}{n^2N_n}$ to bet that each $M_{i,j,k}$ is
	correct, if $M_{i,j,k}$ is incorrect, the martingale loses all its capital
	reserved for betting on $M_{i,j,k}$, otherwise, its capital increases by a
	factor described later.

	We now argue that for each $\pi \in X$, there are infinitely many $n$, such
	that the capital of the martingale after betting on $\binaryn$ is
	$\omega(1)$. Suppose
	that $\pi_n$ is computable by at least one of the first $n$ polynomial time
	TMs on a noticeable subset of $\binaryn$ i.e., $\pi_n$ is computed on a
	subset of size at least $2^n/n^k$ for some positive constant $k$. Let us
	call such a TM, a good TM. Using a simple averaging argument, we can see
	that there must be some size-$n^{\lg n}$ subset $T_i$ where a good TM
	computes $\pi_n$ correctly on a subset of size at least $n^{\lg n}/n^k$.
	For sufficiently large $n$, $n^{\lg n - k} \geq n^{1/2}$. Therefore, for
	sufficiently large $n$, we can always find size-$n^{1/2}$
	subsets where a good TM computes $\pi$. In fact, for sufficiently large $n$,
	there are $\Omega(2^n/n^{\lg n + k})$ subsets, $T_i$, with at least
	$n^{1/2}$ strings that are computed by a good TM.

	We now examine how much capital the martingale gains after betting with
	a good TM. A correct prediction on a string
	$s \in T_i = [0^n + in^{\lg n}, 0^n + (i+1)n^{\lg n})$
	increases the amount bet on $s$ by a factor of at least
	$2^n - (i+1)n^{\lg n} + 1$ for any $i \in [0, 2^n/ n^{\lg n + k})$.
	In the worst case, some $T_i$ with $i = O(2^n/n^{\lg n + k})$ will have
	$n^{1/2}$ strings that are correctly predicted by a good TM. Therefore the
	capital wagered on a the good TM increases by a factor of at least
	\[(2^n - (i+1)n^{\lg n} + 1)^{n^{1/2}}
			= 2^{n^{3/2}} (1 - (i+1)2^{-n}n^{\lg n} + 2^{-n})^{n^{1/2}}
		= \omega({2^{n^{1.4}}}).\]
	The last equality follows because $k > 0$ and $i = O(2^n/n^{\lg n + k})$.

	Since this value grows much faster than
	$\frac{1}{n^2N_n} = \omega(2^{-1.1n})$, the value reserved
	for betting on each $M_{i,j,k}$, we see that our martingale succeeds on $X$.

	Finally, we argue that our martingale operates in linear exponential time.
	The most time consuming step of the martingale is the simulation of $n$
	polynomial time TMs that are run on $N_n = 2^{\Theta(n)}$ subsets with
	$n^{1/2}$ length-$n$ strings. Clearly this can be done in $2^{\Theta(n)}$
	time.
\end{proof}

\begin{corollary}
	If $\pi\in\bfPi$ is $\p$-random, then any polynomial-time $\TM$ will be able to compute $\pi$ on at most a $1/\poly$ fraction for all sufficiently large $n$.
\end{corollary}

Similarly, we can show that random permutations are hard to invert on a
noticeable subset infinitely often.
The main difference is that we search for TMs
inverting the permutation rather than TMs that compute the permutation.

\begin{definition}
	A permutation $\pi\in\bfPi$ is {\em noticeably invertible} if there is a polynomial-time $\TM$ $M$ and a polynomial $p$ such that for infinitely many $n$,
	$| \{ x \in \{0,1\}^n \mid M(\pi(x))=x \} | \geq 2^n/p(n).$
\end{definition}

\begin{theorem}\label{th:noticeably_invertible}
	The set
	$X = \{ \pi \in \bfPi \mid \pi_n \textrm{ is noticeably invertible} \}$
	has $\p$-permutation measure 0.
\end{theorem}

\begin{appendixproof}[Proof of \Cref{th:noticeably_invertible}]
	This proof closely mirrors the argument for Theorem~4.2. The key difference is that we treat the outputs of the TMs as \emph{preimages} rather than images of the permutation. Specifically, given a size-$n^{1/2}$ subset $S = \{x_1, \ldots, x_{n^{1/2}}\} \subseteq \{0,1\}^n$, we run each TM $M$ on $x_i$ and interpret the output $M(x_i)$ (if it is length-$n$) as a prediction that $\pi(M(x_i)) = x_i$.

	The martingale allocates capital to each subset–TM pair and places bets accordingly. It only bets when the outputs $M(x_i)$ are all distinct and of length $n$. As in Theorem~4.2, noticeable invertibility ensures that for infinitely many $n$, there exists a TM making sufficiently many correct inversion predictions. In such cases, the martingale's capital grows superpolynomially, and hence it succeeds.
\end{appendixproof}

\subsection{Random Permutations versus Random Oracles}\label{sec:rand_perm_sw}

Bennett and Gill used random permutations, rather than random languages, to
separate $\P$ from $\NP \cap \coNP$. It is still unknown whether random
oracles separate $\P$ from $\NP \cap \coNP$. In this section, we examine how random permutations yield random languages. We show that a
$\p$-random permutation can be used to generate a $\p$-random language. All of the results in this section are stated for $\p$-randomness. They also hold for $\ptwo$-randomness.

Given a permutation $\pi \in \bfPi$, we define the language
\[L_\pi = \{x \mid \textrm{ the first bit of } \pi(0^{2|x|}x) \textrm{ is 1}\}.\]
For a set of permutations $X \subseteq \bfPi$, we define the set of languages
\[L_X = \{L_\pi \mid \pi \in X\}.\]

\begin{lemma}
	\label{lem:perm_random}
	For any set of permutations $X \subseteq \bfPi$, if a $\p$-computable
	martingale $d$ succeeds on the set of languages
	$L_X = \{L_\pi \mid \pi \in X \}$, then there is a
	$\p$-computable permutation martingale that succeeds on $X$.
\end{lemma}
\begin{appendixproof}[Proof of \Cref{lem:perm_random}]
	We will use $d$ to design a $\p$-computable permutation martingale $D$,
	that succeeds on any $\pi \in X$.

	The main idea behind $D$'s design is to use the capital of the martingale $d$
	to bet on the images of strings that affect the membership of strings in
	$L_\pi$. The bet placed by $D$ on each possible image of the string being
	bet on is proportional to the bet placed by $d$ on the membership in
	$L_\pi$ of the string whose membership is determined by the image.
	For example, the bet placed by $D$ on the image of $0^{2|x|}x$ being $1y$ is
	proportional to the bet placed by $d$ on $x$ being in $L_\pi$.

	The permutation martingale $D$ only bets on strings of form $0^{2n}x$, where
	$x \in \binaryn$ and $n > 0$ (this is the string that determines the
	membership of $x$ in $L_\pi$). Given a string $w$, let
	$\pi\upharpoonright w$ be the prefix list of $\nu_\pi$ that consists of the
	images of all the strings less than $w$. We let $\free(\pi\upharpoonright w)$ be the (nonempty) set of length-$n$ strings that have not yet appeared in the image of the prefix partial permutation $\pi\upharpoonright w$. Where $n$ is the smallest integer for which there exists strings of length $n$ that have not been mapped to any string.
	For $b\in\binary$, $\free(\pi\upharpoonright w)_b$
	is the set of strings of $\free(\pi\upharpoonright w)$ that have $b$ as
	their first bit. Let $F(\pi\upharpoonright w)$ and
	$F(\pi\upharpoonright w)_b$ denote the size of $\free(\pi\upharpoonright w)$
	and $\free(\pi\upharpoonright w)_b$,
	respectively. Finally, let $L(\pi\upharpoonright w)$ be a prefix of the
	characteristic sequence of $L_\pi$, it includes the characteristic bits for
	all strings whose membership can be determined from $\pi\upharpoonright w$.
	For example, if $w = 0^{3n+1}$, then $L(\pi\upharpoonright w)$ is the
	characteristic prefix for the subset $L_\pi$ that consists of all string of
	length at most $n$.

	We now define $D$. As mentioned previously, $D$ only
	bets on strings that have the from $0^{2n}x$, where $x \in \binaryn$ and
	$n > 0$. These are the strings that determine the membership of strings in
	$L_{\pi}$. The martingale $D$ starts with initial capital $d(\lambda)$, the
	initial capital of the martingale that succeeds on $L_\pi$. Without
	loss of generality, we assume that $d$ is non-zero.
	When betting on the images of string $w = 0^{2|x|}x$ with $x \in \binaryn$,
	for each length-$3n$ string beginning with bit $b$, the
	martingale wagers
	$\frac{d(L(\pi\upharpoonright w)b)}{2 d(L(\pi\upharpoonright w))F(\pi\upharpoonright w)_b}$ fraction of its current capital. The term
	$\frac{d(L(\pi\upharpoonright w)b)}{2 d(L(\pi\upharpoonright w))}$ is the
	fraction of $d$'s current capital used to bet on the characteristic bit of
	$x$ in $L_\pi$ being $b$. We then divide this term by
	$F(\pi\upharpoonright x)_b$, the number of length-$3n$ strings that make the
	characteristic bit of $x$ in $L_\pi$ to be $b$.

	It is easy to see that $D$ is computable in $\p$ if $d$ is also computable
	in $\p$. Given $\pi \upharpoonright w$, it takes linear time to compute
	the string $L(\pi\upharpoonright w)$ and the number
	$F(\pi \upharpoonright w)_b$. It takes $\p$-time to compute
	$d(L(\pi\upharpoonright w)b)$ and $2d(L(\pi\upharpoonright w))$. All that is
	left after that is to perform a multiplication and a division, both of which
	requires $\p$-time.

	We now argue that $D$ succeeds on $L_\pi$. Let us examine the
	martingale's capital after betting on string $w = 0^{2|x|}x$. After betting
	on $w$, the martingale's capital multiplied by a factor of
	$\frac{d(L(\pi\upharpoonright w)b)F(\pi\upharpoonright w)}{2 d(L(\pi\upharpoonright w))F(\pi\upharpoonright w)_b}$.
	Since $F_1 = \frac{d(L(\pi\upharpoonright w)b)}{d(L(\pi\upharpoonright w))}$
	represents the factor by which $d$'s capital increases while betting on
	$L_\pi$, we only need to show that the remaining factor
	$F_2 = \frac{F(\pi\upharpoonright w)}{2 F(\pi\upharpoonright w)_b}$ doesn't
	reduce $F_1$ by much, so that if $d$ succeeds on $L_\pi$, $D$ also succeeds
	on $\pi$. We note that $F(\pi\upharpoonright w)$ (the number of
	length-$3n$ strings whose images have yet to be queried) is at least
	$2^{3n}-2^n$, and $F(\pi\upharpoonright w)_b$ (the number of the previously counted strings that begin with bit $b$) is at most $2^{3n-1}$. Therefore,
	$\frac{F(\pi\upharpoonright w)}{2F(\pi\upharpoonright w)_b} \geq \frac{2^{3n}-2^n}{2\times2^{3n-1}} = 1-2^{-2n}$.
	Hence, after betting on all length-$3n$ strings of the
	form $0^{2|x|}x$, $D$'s capital is the same as $d$'s capital after betting
	on length-$n$ strings multiplied by a factor of at least
	$(1-2^{-2n})^{2^n} \approx e^{-2^{-n}}$.  This implies that $D$'s capital
	after betting on length-$3n$ strings grows by approximately the same factor
	by which $d$'s capital grows after betting on length-$n$ strings. Therefore,
	if $d$ succeeds on $L_\pi$, $D$ succeeds on $\pi$.
\end{appendixproof}

\begin{corollary}\label{cor:perm_random}
	If $\pi$ is a $\p$-random permutation, then $L_\pi$ is a $\p$-random language.
\end{corollary}

We now extend the previous lemma to honest $\p$-permutation betting games.
By Lemma \ref{lem:perm_bg_p_i},  honest $\p$-permutation betting games do
not cover $\permE$ and honest $\p_2$-permutation betting games do not cover
$\permEXP$.

\begin{lemma}
	\label{lem:perm_random_1}
	For any set of permutations $X \subseteq \bfPi$, if an honest $\p$-betting
	game $g$ succeeds on the set of languages
	$L_X = \{L_\pi \mid \pi \in X \}$, then there is an honest
	$\p$-permutation betting game $G$ that succeeds on $X$.
\end{lemma}

\begin{appendixproof}[Proof of \Cref{lem:perm_random_1}]
	The proof follows the same strategy as the previous lemma. We use the
	betting game $g$ to construct $G$, a permutation betting game that succeeds
	on $\pi \in X$. The main difference from the previous lemma is in the
	order that
	the strings are queried. The order will be directed by simulating the
	betting game $g$. Whenever $g$ queries any nonempty string $x$, $G$ responds
	by querying the image of $0^{2|x|}x$. The permutation betting game $G$ then bets
	on the first bit of the image of $0^{2|x|}x$ using $g$'s bets on the
	membership of $x$ to determine the proportion placed on the two possible
	outcomes, just like we did in the previous lemma. Since $G$ will not bet on
	strings that do not have the form $0^{2|x|}x$, those strings can be queried
	after $g$ queries all strings of length $n$. Because $g$ is an honest
	$\p$-betting game, it is easy to see that $G$ will be an honest
	$\p$-permutation betting game.

	The arguments for $G$'s capital and runtime are almost identical to that of
	the previous lemmas, the only difference is that we maintain a list of
	strings that may not be in the standard lexicographic order.
\end{appendixproof}

\begin{corollary}\label{co:perm_betting_game_random}
	If $\pi$ is an honest $\p$-betting game random permutation, then $L_\pi$ is an honest $\p$-betting game random language.
\end{corollary}

\begin{definition}
	Given a language $L$, we define $\bfPi_{L}$ to be set of permutations
	\[\bfPi_L = \myset{\pi \in \bfPi }{
			\begin{array}{l}
				\textrm{ for all }  n > 0 \textrm{ and } x \in \binaryn,   \\
				\pi(0^{2n}x) = by \textrm{ for some } y \in \binary^{3n-1} \\
				\textrm{ if and only if } L[x] = b
			\end{array} }.
	\]
	Given a set of languages $X$, we define $\bfPi_X$ as the set of
	permutations $\bfPi_X = \bigcup_{L \in X} \bfPi_{L}.$
\end{definition}

The next lemma is a weak converse to Lemma \ref{lem:perm_random}.

\begin{lemma}\label{lem:perm_random_2}
	For any set of languages $X \subseteq \binaryinfty$, if a $\p$-computable
	permutation martingale $d$ succeeds on the set of permutations
	$\bfPi_X$, then there is a $\p$-computable
	martingale that succeeds on $X$.
\end{lemma}

\begin{appendixproof}[Proof of \Cref{lem:perm_random_2}]
	This proof follows the same general strategy as the proof of
	\Cref{lem:perm_random}. So, we will focus on where the proofs differ.

	We will use the permutation martingale $d$ to construct a martingale $D$,
	that succeeds on the set of languages $X$. The martingale $D$ operates by
	simulating $d$ and using $d$'s bets on strings of the form $0^{2|x|}x$ to
	bet on the membership of $x$ in the language $D$ is betting on.
	Let $L$ be the language $D$ is betting on and let $\pi_L$ be some language
	in $\bfPi_L$. We will specify how $\pi_L$ is constructed when we examine the
	success set of $D$.

	We now analyze how $D$ bets on all nonempty strings; it does not bet on the empty string.
	Let $D_L^x(1)$ denote the portion of $D$'s
	current capital placed on $x \in \binaryn$ being a member of $L$ and
	$D_L^x(0)$ denote the portion placed on $x$ not being a member of $L$. We
	set $D_L^x(1)$ to be the portion $d$ places on the first bit of
	$\pi_L(0^{2n}x)$ being 1 and
	$D_L^x(0)$ to be the portion $d$ places on the first bit of
	$\pi_L(0^{2n}x)$ being 0. For simplicity, we use $w$ to denote the string
	$\pi_L \upharpoonright (0^{2n}x)$. Therefore,
	\[D_L^x(b) = \frac{\sum_{s \in \textrm{free}_b(w)} d(ws)}{|\textrm{free}(w)|d(w)},\]
	where $\textrm{free}_b(w)$ is the set of strings in $\textrm{free}(w)$ that
	begin with bit $b$. Let $s^\prime \in \textrm{free}_b(w)$ be a string that
	minimizes $d(ws)$. Then,
	\[D_L^x(b) \geq
		\frac{|\textrm{free}_b(w)|}{|\textrm{free}(w)|} \frac{d(ws^\prime)}{d(w)}
		\geq
		\frac{(2^{3n-1} - 2^n)}{2^{3n}}\frac{d(ws^\prime)}{d(w)} =  \frac{1}{2}(1-\frac{1}{2^{2n-1}}) \frac{d(ws^\prime)}{d(w)}.\]

	The first inequality follows by our choice of $s^\prime$, the second follows
	because any string of the form $0^{2n}x$ is among the first $2^n$ string of
	$\{0,1\}^{3n}$. After betting that $x$'s characteristic bit is
	$b$, the martingale's capital is multiplied by a factor of
	$2D_L^x(b) \geq (1-\frac{1}{2^{2n-1}}) \frac{d(ws^\prime)}{d(w)}$. As we
	have previously shown, the factor $(1-\frac{1}{2^{2n-1}})$ so small that it
	only contributes a constant factor after betting on all string. So, we only
	have to focus on the factor $\frac{d(ws^\prime)}{d(w)}$. This factor is
	the factor by which $d$'s current capital grows after betting on $x$'s
	characteristic bit being $b$.

	We now specify how $\pi_L$ is constructed.
	When generating the string $w = \pi_L \upharpoonright (0^{2|x|}x)$ which is
	polynomially longer than $L \upharpoonright x$, we choose images of strings
	that do not have the form $0^{|x|}x$ to be any string that minimizes the
	current capital of the permutation martingale, $d$. Now we show that $D$
	succeeds on any language $L \in X$. It is not hard to see that
	$\pi_L \in \bfPi_L \subseteq \bfPi_X$. Therefore, $d$ succeeds
	on $\pi_L$ by the hypothesis of this theorem. By the construction of
	$\pi_L$, the only time $d$'s capital could increase while betting on
	it, is when it is betting on strings of the form $0^{|x|}x$. The factor
	by which $d$'s current capital is multiplied is $\frac{d(ws^\prime)}{d(w)}$,
	the same factor by which $D$ is multiplied after betting on $L[x]$ being
	$b$. Therefore, if $d$ succeeds on $\pi_L$, then $D$ also succeeds on $L$.
	Since $L$ is an arbitrary language in $\bfPi_X$, the theorem follows.
\end{appendixproof}

\begin{corollary}
	If $L$ is a p-random langauge, then $\Pi_L$ does not have p-permutation measure 0.
\end{corollary}

\section{Random Permutations for \texorpdfstring{$\NP$}{NP} \texorpdfstring{$\cap$}{∩} \texorpdfstring{$\coNP$}{coNP}} \label{sec:p_np_conp}

Bennett and Gill \cite{BennettGill81} studied the power of random oracles in separating complexity classes. In particular, they showed that $\P^A \neq \NP^A$ relative to a random oracle with probability 1. However, they were not able to separate $\P$ from $\NP \cap \co\NP$ relative to a random oracle. They also made the observation that if $\P^A = \NP^A \cap \co\NP^A$ for a random oracle $A$, then $\P^A$ must include seemingly computationally hard problems such as factorization. They also proved that any non-oracle-dependent language that belongs to $\P^A$ with probability 1, also belongs to $\BPP$. As a result, if $\P^A = \NP^A \cap \co\NP^A$ for a random oracle $A$ with probability 1, then these difficult problems in $ \NP \cap \coNP$ would be solvable in $\BPP$. To achieve a separation between  $\P^A$ and  $\NP^A \cap \co\NP^A$, they considered length-preserving permutations on $\binary^*$ and showed that
$\P^\pi \neq \NP^\pi \cap \co\NP^\pi$ for every random permutation $\pi$.

Using resource-bounded permutation betting games on the set of all
length preserving permutations of $\binary^\star$, we  strengthen the Bennett-Gill permutation separation, proving that
$\P^\pi\neq \NP^\pi \cap \co\NP^\pi$ for any $\p$-betting-game random
permutations $\pi$. More generally, we show that the set of permutations
$\pi$ such that, $\NP^\pi$ is not $\DTIME^{\pi}(2^{kn})$-bi-immune has
$\p$-permutation-betting-game measure 0. Recall that a language $L$ is bi-immune to a complexity class $C$ if no infinite subset of $L$ or its complement is decidable in $C$ \cite{FlajoletSteyaert74, BalSch85}.

The following is our main theorem where its first part states that relative to a $\p$-betting-game random permutation $\pi$, there is a language $L$ in $\CCfont{NLIN}^\pi \cap \co\CCfont{NLIN}^\pi$ such that  no infinite subset of $L$ or its complement is $\DTIME^\pi(2^{kn})$-decidable.

\begin{theorem}\label{th:main}
	\begin{enumerate}
		\item If $\pi$ is a $\p$-betting-game random permutation, then $\CCfont{NLIN}^\pi \cap \co\CCfont{NLIN}^\pi$ contains a $\DTIME^\pi(2^{kn})$-bi-immune language for all $k \geq 1$.
		\item If $\pi$ is a $\ptwo$-betting-game random permutation, then $\NP^\pi\cap\coNP^\pi$ contains a $\DTIME^\pi(2^{n^k})$-bi-immune language for all $k \geq 1$.
	\end{enumerate}
\end{theorem}
Our headline result is a corollary of \Cref{th:main}.
\begin{corollary}\label{co:main}
	If $\pi$ is a $\p$-betting-game random permutation, then $\P^\pi \neq \NP^\pi \cap \coNP^\pi$.
\end{corollary}

To prove \Cref{th:main}, we first define the following test languages. For each $k \geq 1$, define the ``half range'' test languages

\newcommand{\POLYHALFRANGE}{\CCfont{POLYHRNG}}

\begin{align*}
	\HALFRANGE_{k}^\pi
	 & = \{x\ |\ \exists\ y \in \binary^{k|x|-1},\ \pi(0y) = x^{k}\}      \\
	 & = \{x\ |\ \forall \ y \in \binary^{k|x|-1},\ \pi(1y) \neq x^{k}\},
\end{align*}
and
\begin{align*}
	\POLYHALFRANGE_{k}^\pi
	 & = \{x\ |\ \exists\ y \in \binary^{|x|^k-1},\ \pi(0y) = x^{|x|^{k-1}} \}     \\
	 & = \{x\ |\ \forall \ y \in \binary^{|x|^k-1},\ \pi(1y) \neq x^{|x|^{k-1}}\}.
\end{align*}

A string $x \in \{0,1\}^n$ belongs to $\HALFRANGE_k^\pi$ if the preimage of $x^k$ ($k$ copies of $x$) in $\{0,1\}^{kn}$ begins with 0. If $x$ does not belong to $\HALFRANGE_k^\pi$, then the preimage of $x^k$ begins with 1. In either case, the preimage serves as a witness for $x$. The language $\POLYHALFRANGE_k^\pi$ is similar, but we are looking for a preimage in $\{0,1\}^{n^k}$ of $x^{n^{k-1}}$ ($n^{k-1}$ copies of $x$).
It follows that
\[\HALFRANGE_{k}^\pi \in \NLIN^\pi\cap\co\NLIN^\pi\]and
\[\HALFRANGE_{k}^\pi \in \NTIME^\pi(n^{k}) \cap \co\NTIME^\pi(n^{k})\]
for all $k \geq 1$.

The following lemma implies \Cref{th:main}.

\begin{lemma}\label{le:main_theorem}
	Let $k \geq 0$.
	\begin{enumerate}
		\item The set
		      $X = \{ \pi \in \bfPi \mid \HALFRANGE_{k+3}^\pi \mathrm{\ is\ not\ } \DTIME(2^{kn})^\pi \mathrm{-immune} \}$
		      has $O(n)$-honest $\p$-permutation-betting-game measure 0.
		\item The set
		      $X = \{ \pi \in \bfPi \mid \POLYHALFRANGE_{k+1}^\pi \mathrm{\ is\ not\ } \DTIME(2^{n^{k}})^\pi \mathrm{-immune} \}$
		      has $O(n^{k})$-honest $\ptwo$-permutation-betting-game measure 0.
	\end{enumerate}
\end{lemma}

\begin{appendixproof}[Proof of \Cref{le:main_theorem}]

	We design a betting game that can succeed on two classes of permutations.
	The argument breaks down into two main cases.

	The betting game succeeds on permutations $\pi$ that do \emph{not} map
	strings of the form $1y$ to $x^{k+3}$.
	These permutations are relatively straightforward, because one can easily
	compute the probability that, for all $y$, $\pi(1y) \neq 0^{(k+3)n}$, i.e.\
	$0^n \not\in \HALFRANGE_{k+3}^\pi$.
	This calculation merely requires examining the status of all length-$(k+3)n$
	strings that have been queried.
	In particular, the order in which the betting game queries strings
	\emph{does not} affect its ability to wager on this event, making the
	betting strategy simpler in this case.

	We then consider the complementary class of permutations that \emph{do} map
	some string $1y$ to a string of the form $x^{k+3}$.  Here, $y$ is a
	\emph{potential witness string}.  Unlike the previous case, we exploit
	the betting game's ability to query strings \emph{out of order}.
	In particular, we know there exists some permutation TM capable of predicting
	events of the form $\pi(1y) = x^{k+3}$.  Accordingly, our betting game attempts
	\emph{all} TMs by simulating them on all length-$n$ strings.
	It is precisely this universal simulation that triggers queries
	on (potential) witness strings in arbitrary orders.
	Whenever the simulation would query a witness string, the betting game
	places a bet on that string \emph{first}.

	We show that if, for infinitely many $n$, a witness string is indeed
	queried during the simulation, then the betting game wins unbounded capital.
	On the other hand, if none of the potential witness strings are queried,
	then there is a TM that makes infinitely many \emph{correct} predictions
	of the form $\pi(0y) = x^{k+3}$, thereby allowing another betting strategy
	to succeed.

	We now formalize the above outline and construct the required betting game
	in detail.
	Frist, partition \(X\) into the two sets
	\[
		X_1 \;=\; \{\:\pi \in X \mid \HALFRANGE_{k+3}^\pi
		\text{ is finite}\}\quad\text{and}\quad
		X_2 \;=\; \{\:\pi \in X \mid \HALFRANGE_{k+3}^\pi
		\text{ is infinite}\}.
	\]
	We now design an honest \(\p\)-permutation-betting game, \(G\),
	that succeeds on all permutations in \(X_1\) as well as in \(X_2\).
	To do this, we split \(G\)'s initial capital into infinitely many sub-shares
	according to
	$ a_i = b_i = c_i = \frac{1}{i^2}, i=1,2,\dots.$
	The sub-shares \(\{a_i\}\) are used to bet on permutations in \(X_1\),
	while the sub-shares \(\{b_i,c_i\}\) are reserved for betting on those
	in \(X_2\).

	\textbf{Succeeding on $X_1$.}
	For a randomly selected permutation $\pi$, consider the event $\mathcal{A}_n$, that $0^n \not\in \HALFRANGE_{k+3}^\pi$ i.e., $\pi(1y) = 0^{n(k+3)}$, for some string $y$. The probability of $\mathcal{A}_n$ is 1/2 and it only depends on images of length-$(k+3)n$ strings. To keep track of which length-$(k+3)n$ strings have been queried, we let
	$\omega \in (\,\{\binary^{(k+3)n}\}\;\cup\;\{\star\})^{2^{(k+3)n}}$
	encode the current \emph{status} of each strings. Specifically, if the $i$th length-$(k+3)n$ string has already
	been queried and mapped to some string $z$, then we set
	$\omega[i] = z \in \binary^{(k+3)n}$; otherwise, if it has not
	been queried yet, we set $\omega[i] = \star$.

	Recall that $G$ is our overall betting game, which updates its
	capital shares $a_i$ through a ``subgame'' $G_{a_i}$ whenever $i \ge n$.
	We define
	\[
		G_{a_i}(\omega)
		\;=\;
		\frac{a_i}{\Pr(\mathcal{A}_n)}
		\;\Pr\bigl(\mathcal{A}_n \,\big|\,
		\omega\bigr).
	\]

	Let $\omega^{\,i\to b}$ denote the string $\omega$ with its $i$th component
	replaced by $b\in \{\binary^{(k+3)n}\}\cup\{\star\}$.  We define
	$\free(\omega)$ to be the set of all length-$(k+3)n$ strings available as an image.  It is routine to verify that:
	\[
		G_{a_i}(\omega^{\,i\to \star})
		\;=\;
		\frac{1}{\bigl|\free(\omega^{\,i\to \star})\bigr|}
		\;\sum_{\,b \,\in\, \free(\omega^{\,i\to\star})}\;
		G_{a_i}\bigl(\omega^{\,i\to b}\bigr).
	\]
	Hence, the summation of all $G_{a_i}$ forms a betting game.

	Since $\Pr(\mathcal{A}_n)=\tfrac12$, it is easy to verify that whenever
	$\omega$ represents a configuration in which $\mathcal{A}_n$ is true, we have
	\[
		G_{a_i}(\omega) \;=\; 2\,a_i.
	\]
	Thus every time $\mathcal{A}_n$ occurs, the subgame $G_{a_i}$ doubles its capital
	from $a_i$ to $2\,a_i$.  For every permutation $\pi \in X_1$, $\mathcal{A}_n$
	holds for all but finitely many~$n$, so infinitely many $a_i$'s grow
	unboundedly.  Therefore, $G$ succeeds on all $\pi \in X_1$.

	The ``subgame" $G_{a_i}$ can be computed in $O\bigl(2^{(k+3)n}\bigr)$ time
	by examining the mapping of every string of length~$(k+3)n$.  Moreover,
	the order in which $G$ queries length-\(2n\) strings has no impact on its
	success for $\pi \in X_1$.  As we shall see next, this changes when dealing
	with permutations in~$X_2$.

	\textbf{Succeeding on $X_2$.}
	Given $\pi \in X_2$, we know $\HALFRANGE_{k+3}^\pi$ is not
	$\DTIME(2^{kn})^\pi$-immune.  Thus there is some $\DTIME(2^{kn})^\pi$
	oracle TM that recognizes an infinite subset of $\HALFRANGE_{k+3}^\pi$.
	A string $y \in \{0,1\}^{k(n+1)-1}$ is a \emph{witness} for the membership of
	a length-$n$ string $x$ in $\HALFRANGE_{k+3}^\pi$ if
	$\pi(0y) = x^{(k+1)}$.  We simulate the first $n$ such $\DTIME(2^{kn})^\pi$
	machines on all length-$n$ strings.  Two cases arise:

	(1)~We query witnesses of some $x \in \HALFRANGE_{k+3}^\pi$ for
	infinitely many $n$ during the simulation phase.  Let
	$Q_n \subseteq \{0,1\}^{(k+3)n-1}$ be the set of strings that might
	witness membership of any length-$n$ string $x$ queried in the simulation.
	We use a portion of $b_n$ to bet on the event that $y \in Q_n$ is indeed
	a witness to $x \in \HALFRANGE_{k+3}^\pi$.  For large $n$,
	$|Q_n| < n^2\,2^{(k+1)n}$ because each of the $n$ rounds simulates at most
	$n$ distinct $\DTIME(2^{kn})^\pi$ machines on all length-$n$ inputs.
	We split $b_n$ evenly among the $2^n \times |Q_n|$ possible string-witness
	pairs; a winning bet multiplies that share by at least
	$2^{(k+3)n} - n^2 \,2^{(k+2)n}$ (because the TMs running so far could have
	queried at most $n^2\,2^{(k+2)n}$ length-$(k+3)n$ strings).  Hence the
	capital used per event jumps to
	$ (\frac{b_n}{n\,2^{(k+2)n}}) (2^{(k+3)n} - n^2\,2^{(k+2)n}) = \omega\bigl(2^{0.9n}\bigr),
	$
	and since one of these events arises infinitely often, the betting game
	succeeds.

	(2)~Witness strings are queried only finitely many times.  Eventually, none
	of the TMs ever queries a witness string.  At that stage, after simulating
	the TMs on length-$n$ strings, the betting game employs $c_i$ to wager on
	any $\DTIME(2^{kn})^\pi$ TM $M_i$ that has not made a mistake and accepts
	some $x \in \{0,1\}^n$.  Let $B_{x}$ be the event $x \in
		\HALFRANGE_{k+3}^\pi$, i.e.\ there exists $y$ with
	$\pi(0y)= x^{(k+1)}$.  For a random permutation, $\Pr(B_{x})= \tfrac12$.
	After simulating all TMs on length-$n$ strings \emph{without} querying any
	witness for $x$, the conditional probability of $B_{x}$ remains at least
	$ \frac{2^{(k+3)n} - n^2\,2^{(k+2)n}}{2^{(k+3)n}} \approx \tfrac12.$
	Hence the betting game approximately doubles $c_i$ each time $M_i(x)$ is
	accepted.  Because we are guaranteed that some $M_i$ correctly decides an
	infinite subset of $\HALFRANGE_{k+3}^\pi$, this betting also
	succeeds.

	Clearly, this construction is implementable in $\DTIME(2^{O(n)})$.
	Notice moreover that while querying strings of length $n$, the betting
	game places wagers only on strings of length $(k+3)n$, making it an
	$(k+3)n$-honest $\p$-permutation betting game.

	\textbf{Sketch for part (2).}
	The second statement uses \(\POLYHALFRANGE_{k+1}^\pi\) in place
	of \(\HALFRANGE_{k+3}^\pi\), and the class \(\DTIME(2^{n^k})^\pi\)
	replaces \(\DTIME(2^{kn})^\pi\).  One repeats the same betting-game strategy
	but scales the length parameters and the honesty parameter accordingly;
	the resulting game is \(O(n^{k+1})\)-honest in \(\DTIME(2^{n^{O(1)}})\).  All other details are essentially unchanged. Therefore, the set of permutations
	for which $\POLYHALFRANGE_{k+1}^\pi$ is not
	\(\DTIME(2^{n^k})^\pi\)-immune has $\ptwo$-permutation-betting-game measure~0.

\end{appendixproof}

By symmetry of $\NLIN^\pi\cap\co\NLIN^\pi$ and $\NTIME^\pi(n^{k}) \cap \co\NTIME^\pi(n^{k})$, \Cref{le:main_theorem} also applies to the complement of $\HALFRANGE_{k+3}^\pi$ and $\POLYHALFRANGE_{k+1}^\pi$. Therefore, both languages are bi-immune and \Cref{th:main} follows.

Combining \Cref{le:main_theorem} with \Cref{prop:pspace_random_implies_honest_p_betting_game_random} also gives the following corollary. In the next section we will prove more results about $\pspace$-random permutations.

\begin{corollary}\label{co:main_pspace}
	\begin{enumerate}
		\item If $\pi$ is a $\pspace$-random permutation, then $\CCfont{NLIN}^\pi \cap \co\CCfont{NLIN}^\pi$ contains a $\DTIME^\pi(2^{kn})$-bi-immune language for all $k \geq 1$.
		\item If $\pi$ is a $\ptwospace$-random permutation, then $\NP^\pi\cap\coNP^\pi$ contains a $\DTIME^\pi(2^{n^k})$-bi-immune language for all $k \geq 1$.
	\end{enumerate}
\end{corollary}

\section{Random Permutations for \texorpdfstring{$\NP$}{NP} \texorpdfstring{$\cap$}{∩} \texorpdfstring{$\coNP$}{coNP} versus Quantum Computation}\label{sec:quantum}

Bennett, Bernstein, Brassard, and Vazirani \cite{BeBeBrVa97} showed
that  $\NP^\pi \cap \coNP^\pi \not\subseteq \BQTIME^\pi(o(2^{n/3}))$ relative to a random permutation $\pi$ with probability 1. In this section we investigate how much of their result holds relative to individual random oracles at the space-bounded level.

We begin with a general lemma about test languages and $\QTM$s. We write $\PFPi_{\leq n} = \{ g \in \PFPi \mid |g| \leq 2^{n+1}-1 \}$ for all prefix partial permutations defined on strings in $\{0,1\}^{\leq n}$.  For a string $s_i$ in the standard enumeration, we write $g \restr s_i$ for the length $i$ prefix of $g$. In other words, $g  \restr s_i = [g(s_0),\ldots,g(s_{i-1})]$.

\begin{lemma}\label{le:test_langauge_lemma}
	Let $\pi$ be a permutation with an associated test language $L_{\pi}$ and let $p(n)$ be a linear function (polynomial function, respectively). If for some oracle $\QTM$ $M$ the following conditions hold,  then $\pi$ is not a $\pspace$-random ($\ptwospace$-random, respectively) permutation.
	\begin{enumerateC}
		\item The membership of $0^n$ in $L_{\pi}$ depends on the membership of the strings
		of length at most $p(n)$.
		\item $M^{\pi}$ decides $L_{\pi}$ with  error probability $\delta$, for some constant $0 < \delta < 1$, and queries only  strings of length at most $p(n)$.
		\item For any partial prefix permutation $\rho \in \PFPi_{\leq p(n)}$, the conditional
		probability $$ \PR{|\psi| = p(n)}{M^\psi(0^n) = L_\psi(0^n)\bigm| \rho \prefix \psi}$$
		is computable in
		$O(2^{O(n)})$
		space
		($O(2^{n^{O(1)}})$ space, respectively).
		\item For some constant $1 > \epsilon > \delta$ and for all but  finitely many $n$,
		$$\PR{|\psi| = l(n)}{M^\psi(0^n) = L_\psi(0^n)  \bigm| \pi\restr 0^n \prefix \psi } < 1 - \epsilon.$$
	\end{enumerateC}

\end{lemma}

\begin{appendixproof}[Proof of \Cref{le:test_langauge_lemma}]
	Let $n_0 < n_1 < n_2 < \cdots$ be a sequence of numbers such that $p(n_i) < n_{i+1}$ for all $i \geq 0$. Based on the statement of the lemma, $M^{\psi}(0^{n_{j-1}})$ cannot query any string of length $n_j$ and the membership of $0^{n_{j - 1}}$ does not depend on the membership of any string of length greater than $n_j$. Let $n_k$ be the first number in the sequence above such that for any $n > n_k$ we have $\PR{|\psi| = l(n)}{M^\psi(0^n) = L_\psi(0^n)  \mid \pi\restr 0^n \prefix \psi } < 1 - \epsilon$ and define the following martingale:
	$$d(\rho) =
		\begin{cases}
			1                                                                                                                       & \textrm{if } |\rho| \leq 2^{n_k + 1} - 1                                             \\
			\frac{d(\rho \;\restr\; 0^{n_{j-1}})}{\mathrm{Pr}(\psi \;|\; \rho \;\restr \; 0^{n_{j-1}})}\mathrm{Pr}(\psi \;|\; \rho) & \textrm{if } 2^{n_{j - 1}} < |\rho| \leq 2^{n_j + 1}  - 1 \textrm{, for some } j > k
		\end{cases}
	$$
	where $\mathrm{Pr}(\psi | \nu)$ is the probability that $M^\psi(0^n) = L_\psi(0^n)$ given that $\nu \prefix \psi$. The third condition in the statement of the lemma implies that this martingale is $\pspace$ computable ($\ptwospace$ computable, respectively). Now consider a permutation $\pi$ that satisfies the conditions of the lemma. Then for $j > k$ we have:
	$$
		d(\pi \restr 0^{n_j}) = \frac{d(\pi \restr 0^{n_{j - 1}})}{\mathrm{Pr}(\psi \;|\; \pi \;\restr \; 0^{n_{j-1}})} \mathrm{Pr}(\psi \;|\; \pi \restr 0^{n_j})
		\geq  \frac{d(\pi \restr 0^{n_{j - 1}}) }{1 - \epsilon}(1 - \delta)
	$$

	The last inequality holds because it follows from the first two conditions of the lemma that $\mathrm{Pr}(\psi \;|\; \pi \restr 0^{n_j}) \geq 1 - \delta$. By repeating this process, we can see that $d$ succeeds on $\pi$.

\end{appendixproof}

In the following theorem, we use \Cref{le:test_langauge_lemma} to extend the result by Bennett, Bernstein, Brassard, and Vazirani \cite{BeBeBrVa97} to  $\ptwospace$-random permutations.

\begin{theorem}\label{th:NP_intersect_coNP_BQP}
	If $\pi$ is a $\ptwospace$-random permutation,
	then $\NLIN^\pi \cap \co\NLIN^\pi$ is not contained in $\BQP^\pi$.
\end{theorem}
\begin{appendixproof}[Proof of \Cref{th:NP_intersect_coNP_BQP}]
	Let $M$ be a $\BQP$ oracle machine running in time $t(n)$.
	From Bennett et al. \cite{BeBeBrVa97} it follows that  that relative to a random permutation $\pi$, $M$ fails to decide whether $0^n$ is in $\HALFRANGE^\pi_n$ with probability at least $1/8$.
	We can compute
	$$\PR{|\psi|=t(n)}{M^\psi(0^n) = \HALFRANGE_n^\psi(0^n) \mid \rho \prefix \psi}$$
	in space $2^{n^{O(1)}}$. The Theorem follows from \Cref{le:test_langauge_lemma}.
\end{appendixproof}

We now refine the previous result by considering more restricted quantum  machines that only query strings of $O(n)$ length. This restriction allows us to extend the result to machines with running time $o(2^{n/3})$, analogous to the result of Bennett et al. \cite{BeBeBrVa97}. Whether this extension holds without the restriction on query length remains an open problem.

\begin{theorem}\label{th:NP_intersect_coNP_BQTIME}
	If $\pi$ is a $\pspace$-random permutation and $T(n) = o(2^{n/3})$,
	then $\NLIN^\pi \cap \co\NLIN^\pi$
	is not contained $\BQTIME^{\pi,O(n)\mbox{-}\CCfont{honest}}(T(n))$.
\end{theorem}
\begin{appendixproof}[Proof of \Cref{th:NP_intersect_coNP_BQTIME}]
	Under the honesty condition, we can compute the necessary conditional probability in $2^{O(n)}$ space and apply \Cref{le:test_langauge_lemma}.
\end{appendixproof}
Together, these theorems extend the classical separation of Bennett et al. \cite{BeBeBrVa97} to individual space-bounded random permutations, both in the general and the honest-query setting.

\section{Random Oracles for \texorpdfstring{$\NP$}{NP} \texorpdfstring{$\cap$}{∩} \texorpdfstring{$\coNP$}{coNP} and 0-1 Laws for Measure in \texorpdfstring{$\EXP$}{EXP}}\label{sec:limitations}

Tardos \cite{Tardos89} used the characterizations
\[\BPP = \ALMOST\P =\myset{ A }{ \Pr_R\left[ A \in \P^R\right] = 1} \  \text{\cite{BennettGill81}} \]
and
\[\AM = \ALMOST\NP = \myset{ A }{ \Pr_R\left[ A \in \NP^R\right] = 1}\ \text{\cite{NisWig94}}\]
to prove the following conditional theorem separating $\P$ from $\NPcoNP$ relative to a random oracle.
\begin{theorem_cite}{Tardos \cite{Tardos89}}\label{th:Tardos}
	If $\AM \cap \coAM \not= \BPP$, then $\P^R \not= \NP^R \cap
		\coNP^R$ for a random oracle $R$ with probability 1.
\end{theorem_cite}

Recently, Ghosal et al. \cite{Ghosal:STOC23} used non-interactive zero-knowledge (NIZK) proofs to prove a similar conditional theorem.

\begin{theorem_cite}{Ghosal et al. \cite{Ghosal:STOC23}}\label{th:Ghosal}
	If $\UP \not\subseteq \RP$,
	then $\P^R \neq \NP^R \cap \coNP^R$ for a random oracle $R$ with probability 1.
\end{theorem_cite}

In this section we use \Cref{th:Tardos,th:Ghosal} to connect the open problem of $\P$ versus $\NPcoNP$ relative to a random oracle to open questions about the resource-bounded measure of complexity classes within $\EXP$. In particular, we relate the problem to measure 0-1 laws and measurability in $\EXP$. First, we need the following derandomization lemma. The first two parts follow from previous work, while the third part of the lemma is a new observation as far as we know, though its proof uses the techniques from the proofs of the first two parts.

\begin{lemma}\label{le:measure_derand}
	\begin{enumerate}
		\item If $\mup(\NP) \neq 0$, then $\BPP \subseteq \NPcoNP = \AMcoAM$.
		\item If $\mup(\UPcoUP) \neq 0$, then $\BPP \subseteq \UPcoUP$.
		\item If $\mup(\FewP) \neq 0$, then $\BPP \subseteq \FewP \cap \co\FewP$.
	\end{enumerate}
\end{lemma}
\begin{appendixproof}[Proof of \Cref{le:measure_derand}]
	\begin{enumerate}
		\item  $\mup(\NP) \neq 0$ implies $\NP = \AM$ \cite{ImpMos09}. Then $\BPP \subseteq \AMcoAM = \NPcoNP$.

		\item If $\mup(\UPcoUP) \neq 0$, then the $\UP$-machine hypothesis holds, which implies $\BPP \subseteq \UPcoUP$ \cite{Hitchcock:HHDCC}.

		\item If $\mup(\FewP) \neq 0$, then there is a $\p$-random language $R$ in $\FewP$ \cite{AmTeZh97}. All but finitely many witnesses for membership in $R$ have high circuit complexity \cite{ImpMos09,Hitchcock:HHDCC}. By the longest runs theorem \cite{Harkins:SRICS}, there is guaranteed to be a string in $R$ within the first $2n$ strings of $\{0,1\}^n$, for all but finitely many $n$. We guess one of these strings and a witness. If we find a valid witness, we use it to build a pseudorandom generator \cite{NisWig94,ImpWig01} and derandomize $\BPP$ \cite{Hitchcock:HHDCC}. There are at most a polynomial number of witnesses, placing $\BPP \subseteq \FewP \cap \co\FewP$.
	\end{enumerate}
\end{appendixproof}

In the following theorem, we have three hypotheses where a complexity class $X$ is assumed to be not equal to $\EXP$ and the $\p$-measure of a subclass of $X$ is concluded to be 0.

\begin{theorem}\label{th:limitations_main}
	Suppose that $\P^R = \NP^R\cap\coNP^R$ for a random oracle $R$ with probability 1. Then all of the following hold:
	\begin{enumerate}
		\item $\NP \neq \EXP \Rightarrow \mup(\NP \cap \coNP) = 0.$
		\item $\UP \neq \EXP \Rightarrow \mup(\UP \cap \coUP) = 0.$
		\item $\FewP \neq \EXP \Rightarrow \mup(\UP) = 0.$
	\end{enumerate}
\end{theorem}
\begin{appendixproof}[Proof of \Cref{th:limitations_main}]
	\begin{enumerate}
		\item Suppose $\mup(\NPcoNP) \neq 0$ and $\NP \neq \EXP$. From \Cref{le:measure_derand} we have $\BPP \subseteq \NPcoNP = \AMcoAM$. Therefore $\BPP \subseteq \NP \neq \EXP$, so $\mup(\BPP) =0$ by the zero-one law for $\BPP$ \cite{vanM00}. Since $\AMcoAM$ and $\BPP$ have different $\p$-measures, the classes are not equal. The result follows from \Cref{th:Tardos}.
		\item Suppose $\mup(\UPcoUP) \neq 0$ and $\UP \neq \EXP$. From \Cref{le:measure_derand} we have $\BPP \subseteq \UPcoUP$.
		      Therefore $\BPP \subseteq \UP \neq \EXP$, so $\mup(\BPP) = 0$ by the zero-one law. Therefore $\UPcoUP \not\subseteq \BPP$, so $\UP \not\subseteq \RP$. The result follows from \Cref{th:Ghosal}.
		\item Suppose $\mup(\UP) \neq 0$ and $\FewP \neq \EXP$. Since $\UP \subseteq \FewP$, we have $\mup(\FewP) \neq 0$ and \Cref{le:measure_derand} implies $\BPP \subseteq \FewP \cap \co\FewP$. Therefore $\BPP \subseteq \FewP \neq \EXP$, so $\mup(\BPP) = 0$ by the zero-one law. Therefore $\UP \not\subseteq \BPP$, so $\UP \not\subseteq \RP$. The result follows from \Cref{th:Ghosal}.
	\end{enumerate}
\end{appendixproof}

\Cref{th:limitations_main} has the following corollary about measure 0-1 laws in $\EXP$. We recall the definitions $\mu(X \mid \EXP) = 0$ if $\muptwo(X \cap \EXP) = 0$ and $\mu(X \mid \EXP) = 1$ if $\muptwo(X^c \mid \EXP) = 0$ \cite{Lutz:AEHNC}.

\begin{corollary}\label{co:limitations_zero_one_laws}
	Suppose that $\P^R = \NP^R\cap\coNP^R$ for a random oracle $R$ with probability 1. Then all of the following hold:
	\begin{enumerate}
		\item $\mu(\NP \cap \coNP \mid \EXP) \in \{0,1\}.$
		\item $\mu(\UP \cap \coUP \mid \EXP) \in \{0,1\}.$
		\item
		      $\mu(\UP \mid \EXP) = 0$
		      or   $\mu(\FewP \mid \EXP) = 1.$
	\end{enumerate}
\end{corollary}

\begin{appendixproof}[Proof of \Cref{co:limitations_zero_one_laws}]
	The two key facts we need in this proof are that for any class $X \subseteq \EXP$:
	\begin{enumerate}[\upshape (a)]
		\item If $X$ is closed under finite union and intersection, then $X = \EXP$ if and only if $\mu(X \mid \EXP) = 1$ \cite{ReSiCa95,Lutz:QSET}.
		\item If $\mup(X) = 0$, then $\mu(X \mid \EXP) = 0$ \cite{Lutz:AEHNC}.
	\end{enumerate}
	We write each implication in \Cref{th:limitations_main} as a disjunction and apply the above facts:
	\begin{enumerate}
		\item We have $\NP = \EXP \Leftrightarrow \NPcoNP = \EXP \Leftrightarrow \mu(\NPcoNP\mid \EXP) = 1$ or $\mup(\NPcoNP) = 0 \Rightarrow \mu(\NPcoNP \mid \EXP) = 0$.
		\item We have $\UP = \EXP \Leftrightarrow \UPcoUP = \EXP \Leftrightarrow \mu(\UPcoUP \mid \EXP) = 1$ or $\mup(\UPcoUP) = 0 \Rightarrow \mu(\UPcoUP\mid\EXP) = 0.$
		\item We have $\FewP = \EXP \Leftrightarrow \mu(\FewP \mid \EXP) = 1$ or $\mup(\UP) = 0 \Rightarrow \mu(\UP \mid \EXP) = 0$.
	\end{enumerate}
\end{appendixproof}
In the third case of \Cref{co:limitations_zero_one_laws}, we almost have a 0-1 law for $\UP$. Can a full 0-1 law be obtained?

The contrapositives of the implications in \Cref{co:limitations_zero_one_laws} show that the random oracle question for $\P$ versus $\NPcoNP$ is resolved under nonmeasurability hypotheses. A complexity class $X$ is defined to be {\em not measurable} in $\EXP$ if $\mu(X\mid\EXP) \neq 0$ and $\mu(X\mid\EXP) \neq 1$ \cite{Lutz:QSET,ReSiCa95}.
\begin{corollary}\label{co:limitations_nonmeasurability}
	\begin{enumerate}
		\item If $\NPcoNP$ is not measurable in $\EXP$, then $\P^R \neq \NP^R\cap\coNP^R$ for a random oracle $R$ with probability 1.
		\item If $\UPcoUP$ is not measurable in $\EXP$, then $\P^R \neq \NP^R\cap\coNP^R$ for a random oracle $R$ with probability 1.
		\item If $\UP$ and $\FewP$ are both not measurable in $\EXP$, then $\P^R \neq \NP^R\cap\coNP^R$ for a random oracle $R$ with probability 1.
	\end{enumerate}
\end{corollary}

On the other hand, if the consequence of \Cref{co:limitations_nonmeasurability} can be proved with measure in $\EXP$, then we would have $\BPP \neq \EXP$, which implies $\mu(\BPP \mid \EXP) = 0$ by the 0-1 law for $\BPP$ \cite{vanM00}.
\begin{theorem}\label{th:ptwo_measure_BPP_neq_EXP}
	If $\{ A \mid \P^A = \NP^A \cap \coNP^A \}$ has measure 0 in $\EXP$, then $\mu(\BPP \mid \EXP) = 0$.
\end{theorem}
\begin{appendixproof}[Proof of \Cref{th:ptwo_measure_BPP_neq_EXP}]
	Let $A$ be $\leqpT$-complete for $\EXP$.
	Then
	$$\NP^A \cap \coNP^A \subseteq \EXP \subseteq \P^A \subseteq \NP^A \cap \coNP^A.$$
	Therefore the $\leqpT$-complete sets are a subset of $\{ A \mid \P^A = \NP^A \cap \coNP^A \}$, so they have measure 0 in $\EXP$, which implies
	$\BPP \neq \EXP$ \cite{BvMRSS01} and $\mu(\BPP \mid \EXP) = 0$.
\end{appendixproof}

These results suggest that resolving whether $\P^R = \NP^R \cap \coNP^R$ relative to a random oracle $R$ requires a deeper understanding of the resource-bounded measurability within $\EXP$ of fundamental subclasses such as $\BPP$, $\NP$, $\UP$, and $\FewP$.

\section{Conclusion}\label{sec:conclusion}

We have introduced resource-bounded random permutations and shown that
$\P^\pi \neq \NP^\pi\cap\coNP^\pi$ for all $\p$-betting-game random permutations. We remark that all of the results in \Cref{sec:p_np_conp,sec:quantum} about $\NLIN\cap\co\NLIN$ and $\NP\cap\co\NP$  hold for their unambiguous versions $\CCfont{ULIN}\cap\co\CCfont{ULIN}$ and $\UPcoUP$, respectively.
An interesting open problem is whether our main theorem can be improved from betting-game random permutations to random permutations.

\begin{question}
	Does $\P^\pi \neq \NP^\pi \cap \co\NP^\pi$ for a $\p$-random permutation $\pi$?
\end{question}

More generally, the relative power of permutation martingales versus betting games should be investigated.

\begin{question}
	Are polynomial-time permutation martingales and permutation betting games equivalent?
\end{question}

We proved two restricted versions of the Bennett et al.~\cite{BeBeBrVa97} random permutation separation. Does the full version hold relative to individual random permutations?

\begin{question}\label{conj:NP_intersect_coNP_BQTIME}
	If $\pi$ is a $\pspace$-random permutation and $T(n) = o(2^{n/3})$,
	is $\NLIN^\pi \cap \co\NLIN^\pi$ not contained in $\BQTIME^\pi(T(n))$?
\end{question}

%


\end{document}